\documentclass[showpacs,twocolumn,aps,pra,longbibliography,superscriptaddress,notitlepage]{revtex4-2}
\usepackage{qcircuit}
\usepackage[dvips]{graphicx}
\usepackage{amsmath,amssymb,amsthm,mathrsfs,amsfonts,dsfont}
\usepackage{subfigure, epsfig}
\usepackage{braket}
\usepackage{bm}
\usepackage{enumerate}
\usepackage{algorithm}
\usepackage{algpseudocode}
\usepackage{diagbox}
\usepackage{physics}
\usepackage{color}
\usepackage{multirow}
\usepackage[marginal]{footmisc}
\usepackage{comment}
\usepackage{tikz}
\usepackage[]{qcircuit}
\usepackage{makecell}
\usetikzlibrary{arrows}
\usetikzlibrary{shapes,fadings,snakes}
\usetikzlibrary{decorations.pathmorphing,patterns}
\usetikzlibrary{calc}
\usetikzlibrary{positioning}
\usepackage[colorlinks = true]{hyperref}

\graphicspath{{./figure/}}

\newtheorem{theorem}{Theorem}

\newtheorem{fact}{Fact}
\newtheorem{lemma}{Lemma}

\newtheorem{definition}{Definition}

\newcommand{\mc}{\mathcal}

\newcommand{\mbb}{\mathbb}

\newcommand{\comments}[1]{}

\hypersetup{
colorlinks=true,
linkcolor=blue,
filecolor=blue,
citecolor=blue,  
urlcolor=blue,
}

\begin{document}
\newcommand{\bignorm}[1]{\big\| #1 \big\|}
\newcommand{\Bignorm}[1]{\Big\| #1 \Big\|}

\title{BELT: Block Encoding of Linear Transformation on Density Matrices}
\begin{abstract}
Linear maps that are not completely positive play a crucial role in the study of quantum information, yet their non-completely positive nature renders them challenging to realize physically. The core difficulty lies in the fact that when acting such a map $\mc{N}$ on a state $\rho$, $\mc{N}(\rho)$ may not correspond to a valid density matrix, making it difficult to prepare directly in a physical system.
We introduce Block Encoding of Linear Transformation (BELT), a systematic protocol that simulates arbitrary linear maps by embedding the output $\mc{N}(\rho)$ into a block of a unitary operator. BELT enables the manipulation and extraction of information about $\mc{N}(\rho)$ through coherent quantum evolution. Notably, BELT accommodates maps that fall outside the scope of quantum singular value transformation, such as the transpose map.
BELT finds applications in entanglement detection, quantum channel inversion, and simulating pseudo-differential operators, and demonstrates improved sample complexity compared to protocols based on Hermitian-preserving map exponentiation.
\end{abstract}

\date{\today}

\author{Fuchuan Wei}
\affiliation{Yau Mathematical Sciences Center, Tsinghua University, Beijing 100084, China}
\affiliation{Department of Mathematics, Tsinghua University, Beijing 100084, China}

\author{Rundi Lu}
\affiliation{Yau Mathematical Sciences Center, Tsinghua University, Beijing 100084, China}
\affiliation{Department of Mathematics, Tsinghua University, Beijing 100084, China}

\author{Yuguo Shao}
\affiliation{Yau Mathematical Sciences Center, Tsinghua University, Beijing 100084, China}
\affiliation{Department of Mathematics, Tsinghua University, Beijing 100084, China}

\author{Junfeng Li}
\affiliation{Yau Mathematical Sciences Center, Tsinghua University, Beijing 100084, China}
\affiliation{Department of Mathematics, Tsinghua University, Beijing 100084, China}

\author{Jin-Peng Liu}
\email{liujinpeng@tsinghua.edu.cn}
\affiliation{Yau Mathematical Sciences Center, Tsinghua University, Beijing 100084, China}
\affiliation{Yanqi Lake Beijing Institute of Mathematical Sciences and Applications, Beijing 100407, China}

\author{Zhengwei Liu}
\email{liuzhengwei@mail.tsinghua.edu.cn}
\affiliation{Yau Mathematical Sciences Center, Tsinghua University, Beijing 100084, China}
\affiliation{Department of Mathematics, Tsinghua University, Beijing 100084, China}
\affiliation{Yanqi Lake Beijing Institute of Mathematical Sciences and Applications, Beijing 100407, China}

\maketitle

\section{Introduction}

A legitimate quantum operation must map a density matrix to another density matrix, requiring it to be completely positive and trace-preserving (CPTP) \cite{nielsen2010quantum,watrous2018theory}.
However, linear maps that are not completely positive (non-CP) frequently arise in many quantum information contexts \cite{Pechukas1994reduced,gunhe2009entanglement,cai2022quantum,endo2021hybrid,wei2023realizing,Son2025Dynamic}.
In entanglement detection~\cite{gunhe2009entanglement}, entanglement can be determined by the positivity of output matrices generated by applying a positive but non-CP map to subsystems.
In scenarios when the inverse of a quantum channel should be performed, e.g., quantum error mitigation \cite{cai2022quantum,endo2021hybrid}, the inverse process becomes non-CP if the original channel is non-unitary. 
In non-Markovian quantum dynamics \cite{Pechukas1994reduced, Salgado2004evolution, carteret2008dynamics}, correlations between the input state and the environment cause the non-CP processes.

Simulating non-CP maps in physical protocols becomes a crucial task in advanced studies. Quasiprobability sampling \cite{temme2017mitigation,endo2018practical,Zhao2025Power} is popularly utilized to simulate non-CP maps by expressing the target non-CP map as a quasiprobability mixture of several CPTP maps, and retrieving the desired expectation value of the output of the non-CP map at the cost of some sampling overhead. However, such a method cannot prepare the output state of the non-CP map, and thus it has limited power when dealing with tasks such as sampling from the output state, quantum communication, and quantum storage.
Other methods like structural approximation \cite{horodecki2002direct,korb2008structural} and $N$-copy extension \cite{dong2019positive} are designed for some special kinds of non-CP maps and lack generality.

As a versatile tool and a building-block of many quantum algorithms \cite{Harrow2009HHL,gilyn2019singular,Dong_2021,Costa2022optimal,Gilyen2022Petz,Low2024Eigenvalue,chakraborty2025quantumsingularvaluetransformation,niwa2025singularvaluetransformationunknown}, the technique of block encoding \cite{Low2017optimal} embeds a matrix $A$ as a submatrix or ``block" within a larger, unitary matrix $U_A$, which is realizable by a quantum computer.
This approach leverages the structure of quantum computing to indirectly access and manipulate the properties of the nonunitary matrix $A$, leading to various non-unitaries processes with improved computational efficiency.
Given access to the unitary that prepares the purification of a state $\rho$, and the inverse of this unitary, one can obtain $\rho$'s block encoding, with applications in Hamiltonian simulation \cite{Low2019hamiltonian,gilyn2019singular} and performing the Petz recovery map  \cite{Gilyen2022Petz}. 
Recent work pushes block encoding beyond matrices to time-evolution operators. Both the linear-combination-of-Hamiltonian-simulation (LCHS) framework \cite{ALL23,ACL23,ACLY24,LLLL25} and the Schrödingerisation technique \cite{JLY23,JLLY23,JLLY24} offer what can be viewed as a continuous block encoding, embedding an entire evolution inside a smoothly parameterized unitary family. Together, these complementary approaches elevate block encoding to a unified toolkit for simulating non-Hermitian quantum dynamics and time-dependent linear differential equations.

\begin{figure*}[t]
\centering
\includegraphics[width=0.95\textwidth]{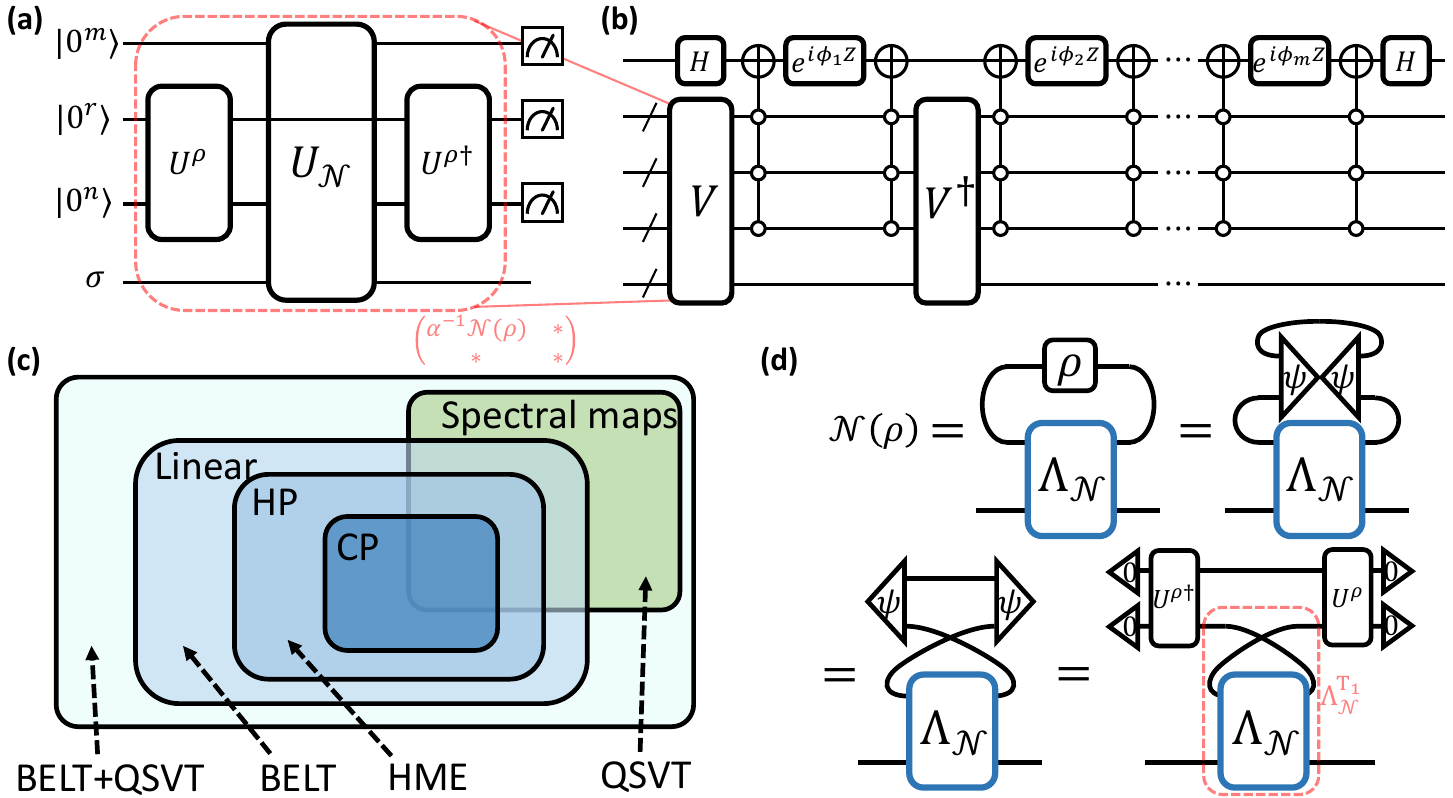}
\caption{(a) Circuit for BELT, which yields an $(\alpha,m+r+n,\epsilon)$-block encoding of $\mc{N}(\rho)$. Here $U_{\mc{N}}$ is an $(\alpha,m,\epsilon)$-block encoding of $\Lambda_{\mc{N}}^{\mathrm T_1}$. The circuit in (a) can be inserted into the standard QSVT circuit in (b) to obtain a block encoding of $f(\alpha^{-1}\mc{N}(\rho))$, where $f:\mbb{R}\rightarrow\mbb{C}$ is a spectral function. For instance, choosing $f$ to be a Chebyshev polynomial implements oblivious amplitude amplification.
(c) Relationships among classes of transformations acting on density matrices in $L(\mc{H})$. ``Linear" denotes the set $L(L(\mc{H}))$. Linear maps can be simulated by BELT (Theorem~\ref{thm:main}). ``HP" denotes Hermitian-preserving linear maps, which can be simulated by the HME algorithm~\cite{wei2023realizing}. ``CP" denotes completely positive linear maps. ``Spectral maps" denote transformations $f:\mbb{R}\rightarrow\mbb{C}$ that act on singular values and are realizable by QSVT. Combining BELT with QSVT enables a broader family of transformations, e.g., the map $T_N \circ \bignorm{\Lambda_{\mc{E}^{-1}}^{\mathrm T_1}}_{\infty}^{-1}\mc{E}^{-1}$ used in Theorem~\ref{thm:inverting} to invert $\mc{E}$.
(d) Tensor network proof of Theorem~\ref{thm:main}. Matrices are depicted as boxes with left (row) and right (column) legs, and vectors/dual vectors are depicted as triangles.
}
\label{fig:topfigure}
\end{figure*}

In this work, to overcome the restriction that non-CP maps generally do not produce valid density matrices,
we introduce Block Encoding of Linear Transformation (BELT). This quantum algorithm can simulate an arbitrary linear map $\mc{N}$ acting on an input state $\rho$, by block encoding the output $\mc{N}(\rho)$, thereby transforming the non-physical object $\mc{N}(\rho)$ into a physically realizable unitary.
BELT generalizes the technique of block encoding $\rho$ by replacing the swap operator with a unitary that encodes the action of $\mc{N}$.
We demonstrate several applications of BELT. 
First, when $\mathcal{N}$ is a positive but not completely positive map, BELT enables an efficient entanglement-detection protocol, whose sample complexity shows a constant-versus-exponential separation compared to any single-copy protocol~\cite{wei2023realizing,liu2024separation}.  
Second, by choosing $\mathcal{N} = \mathcal{E}^{-1}$ to be the inverse map of a CPTP map $\mc{E}$, BELT implements the channel inversion, allowing one to reconstruct an unknown state $\psi$ from multiple copies of $\mathcal{E}(\psi)$.
Third, interpreting $\mathcal{N}$ as a non-unitary evolution operator lets BELT simulate non-Hermitian dynamics; in particular, we showcase its ability to realize pseudo-differential operators.
Compared to the Hermitian-preserving map exponentiation (HME) algorithm~\cite{wei2023realizing}, BELT extends the scope of simulable maps to include those that are not Hermitian-preserving, while also offering improved sample complexities in the tasks mentioned above (see Table~\ref{tab:HME_BE_comparison}).

We recall that quantum singular value transformation (QSVT) \cite{gilyn2019singular,martyn2021grand} enables an approximate block-encoding of $f(\rho)$, where $f:\mathbb{R}\to\mathbb{C}$ acts on the spectrum of $\rho$.
However, general linear maps are not, in general, spectral functions of density operators. 
For example, the transpose map $\rho\mapsto\rho^{\mathrm T}$ is basis dependent and cannot be written as $f(\rho)$ for any scalar function $f$.
Consequently, BELT is distinct from QSVT and can simulate transformations that standard QSVT-based methods cannot access.
Used together, BELT and QSVT provide complementary capabilities, namely linear maps and spectral functions, thereby enlarging the class of implementable transformations (see Fig.~\ref{fig:topfigure}(c)).


\section{Block encoding of linear transformation}

When a non-CP linear map acts on a quantum state, the output state may lose positivity and hence fail to represent a valid quantum state.  This obstacle prevents a direct physical implementation of non-CP maps, as we cannot physically prepare their output.  
To avoid the non-physical nature of non-CP maps, one possible idea is to change the carrier of the output state.
Specifically, instead of preparing $\mc{N}(\rho)$ itself, we encode it coherently as a submatrix of an enlarged unitary operator.

In this section, we show how to construct such a unitary for an arbitrary but known linear map $\mc{N}$ (which may be non-CP) and an unknown input state $\rho$.  Throughout, let $L(\mc{H})$ denote the space of linear operators on a Hilbert space $\mc{H}$, $\mbb{I}_n$ the identity on $n$ qubits, and $\|\cdot\|_\infty$ the operator norm (i.e., the largest eigenvalue).  We begin by recalling the formal definition of block encoding:

\begin{definition}[Block encoding \cite{lin2022lecture}]
For an $n$-qubit matrix $A$, if there exists $\alpha,\epsilon\in\mbb{R}_{\ge0}$, and an $(m+n)$-qubit unitary $U\in L(\mc{H}'\otimes\mc{H})$ such that
\begin{equation}
\bignorm{A-\alpha(\bra{0^m}\otimes\mbb{I}_{n})U(\ket{0^m}\otimes\mbb{I}_{n})}_{\infty}\le\epsilon,
\end{equation}
we say the unitary $U$ is an $(\alpha,m,\epsilon)$-block encoding of $A$.
\end{definition}

A block encoding of $A$ exists iff $\|A\|_\infty\le\alpha$.  
When the unitary $U$ is expressed in the computational basis, the upper-left block of $U$ equals $A/\alpha$ up to $\epsilon$ error.

For $\mc{N}:L(\mc{X})\rightarrow L(\mc{Y})$ a linear map, let $\Lambda_{\mc{N}}=(\mc{I}\otimes\mc{N})(\Phi^{+})$,
where $\ket{\Phi^{+}}=\sum_{i}\ket{ii}$ is the unnormalized maximally entangled state on $\mc{X}\otimes\mc{X}$.
The matrix $\Lambda_{\mc{N}}$ is called the Choi matrix of $\mc{N}$ and contains all the information of the map $\mc{N}$.

\begin{theorem}[BELT]\label{thm:main}
Let $\mc N : L(\mathbb C^{2^{n}})\to L(\mathbb C^{2^{k}})$ be a linear map.
Let $U^{\rho}$ be an oracle that prepares a purification of an $n$-qubit state $\rho$, i.e.\ $U^{\rho}\ket{0^{r+n}}=\ket{\psi}$ with
$\Tr_{\mathbb C^{2^{r}}}\ket{\psi}\bra{\psi}=\rho$.
Assume $U_{\mc N}$ is an $(\alpha,m,\epsilon)$-block encoding of the
partially transposed Choi matrix
$\Lambda_{\mc N}^{\mathrm T_1}$, where ${\mathrm T_1}$ denotes the partial transpose on the first subsystem.
Then
\begin{equation}\label{eq:circuit}
  \bigl(\mathbb I_m\otimes U^{\rho\dagger}\otimes\mathbb I_k\bigr)
  \bigl(U_{\mc N}\otimes\mathbb I_r\bigr)
  \bigl(\mathbb I_m\otimes U^{\rho}\otimes\mathbb I_k\bigr)
\end{equation}
is an $(\alpha,m+r+n,\epsilon)$-block encoding of $\mc N(\rho)$.
\end{theorem}

The full proof is given in Appendix~\ref{app:proof_of_main}; below we sketch the $\epsilon=0$ case.
Applying the tensor-network identities in Fig.~\ref{fig:topfigure}(d) yields
\begin{equation}\label{eq:sketch_main}
  \mc N(\rho)=\bigl(\bra{0^{r+n}}U^{\rho\dagger}\otimes\mathbb I_k\bigr)
    \bigl(\mathbb I_r\otimes\Lambda_{\mc N}^{\mathrm T_1}\bigr)
    \bigl(U^{\rho}\ket{0^{r+n}}\otimes\mathbb I_k\bigr).
\end{equation}
Because
$\Lambda_{\mc N}^{\mathrm T_1} = \alpha(\bra{0^{m}}\otimes\mathbb I_{n+k})U_{\mc N}(\ket{0^{m}}\otimes\mathbb I_{n+k})$,
the only non-unitary element $\Lambda_\mc{N}^{\mathrm{T}_1}$ in Eq.~\eqref{eq:sketch_main} can be replaced by the unitary $U_{\mc N}$, establishing that Eq.~\eqref{eq:circuit} is indeed an $(\alpha,m+r+n,0)$-block encoding of $\mc N(\rho)$.

The circuit corresponding to Eq.~\eqref{eq:circuit} is shown in Fig.~\ref{fig:topfigure}(a).
Denote
$ V :=
  \bigl(\mathbb I_m\otimes U^{\rho\dagger}\otimes\mathbb I_k\bigr)
  \bigl(U_{\mc N}\otimes\mathbb I_r\bigr)
  \bigl(\mathbb I_m\otimes U^{\rho}\otimes\mathbb I_k\bigr)$.
By Theorem~\ref{thm:main}, for $\epsilon=0$, inputting the state
$\ketbra{0^{m+r+n}}{0^{m+r+n}}\otimes\sigma$
and post-selecting on outcome $\ket{0^{m+r+n}}$ on the first $m+r+n$ qubits produces the state
\begin{equation}
\frac{\mc N(\rho)\sigma\mc N(\rho)^{\dagger}}
       {\Tr\bigl[\mc N(\rho)\sigma\mc N(\rho)^{\dagger}\bigr]}
\end{equation}
on the final $k$ qubits. The success probability is
\begin{align}
&\Pr\bigl[\text{get }\ket{0^{m+r+n}}\bigr]=\\
&\Tr\bigl[
     \bra{0^{m+r+n}}\!V\! \ket{0^{m+r+n}}
     \sigma
     \bra{0^{m+r+n}}\!V^{\dagger}\! \ket{0^{m+r+n}}
    \bigr]\\
&= \alpha^{-2}
    \Tr\bigl[\mc N(\rho)\sigma\mc N(\rho)^{\dagger}\bigr].
\end{align}
Since $U_{\mc N}$ is an $(\alpha,m,\epsilon)$-block encoding, we have
$\alpha\ge\|\Lambda_{\mc N}^{\mathrm T_1}\|_\infty$; hence the success probability of obtaining $\ket{0^{m+r+n}}$ is bounded above by
$\|\Lambda_{\mc N}^{\mathrm T_1}\|_\infty^{-2}\Tr\bigl[\mc N(\rho)\sigma\mc N(\rho)^{\dagger}\bigr]$.
Therefore, when identifying applications of BELT, it is preferable to consider maps $\mc N$ for which the quantity $\|\Lambda_{\mc N}^{\mathrm T_1}\|_\infty^{-2}$ remains bounded or grows only moderately with system size, as demonstrated in the following sections.
With a suitable choice of $\sigma$, the probability $\alpha^{-2}\Tr\bigl[\mc N(\rho)\sigma\mc N(\rho)^{\dagger}\bigr]$ reveals information about $\mc N(\rho)$. When $\mc N(\rho)$ is pure, the post-selected state $\frac{\mc{N}(\rho)\sigma \mc{N}(\rho)^{\dagger}}{\Tr(\mc{N}(\rho)\sigma \mc{N}(\rho)^{\dagger})}
 =\mc{N}(\rho)$, enabling its preparation whenever the post-selection succeeds.

For the identity channel $\mc N=\mathcal I$, one has
$\Lambda_{\mathcal I}^{\mathrm T_1}=S$, the swap operator on $n+n$ qubits.
Taking $m=0$ and $U_{\mc N}=S$ reduces Eq.~\eqref{eq:circuit} to the standard block encoding of~$\rho$ \cite{Low2019hamiltonian,gilyn2019singular}.
BELT replaces $S=\Lambda_{\mathcal I}^{\mathrm T_1}$ with $\Lambda_{\mathcal N}^{\mathrm T_1}$, thereby applying $\mc N$ to the block-encoded $\rho$.

The block-encoded matrix $\alpha^{-1}\mc N(\rho)$ can be processed with QSVT \cite{gilyn2019singular,martyn2021grand}: the circuit in Fig.~\ref{fig:topfigure}(b) (approximately) implements a real function $f:\mathbb R\to\mathbb R$ on its singular values, yielding a block-encoded matrix $f(\alpha^{-1}\mc{N}(\rho))$.

\section{Applications}

\begin{figure*}[t]
\centering
\includegraphics[width=0.94\textwidth]{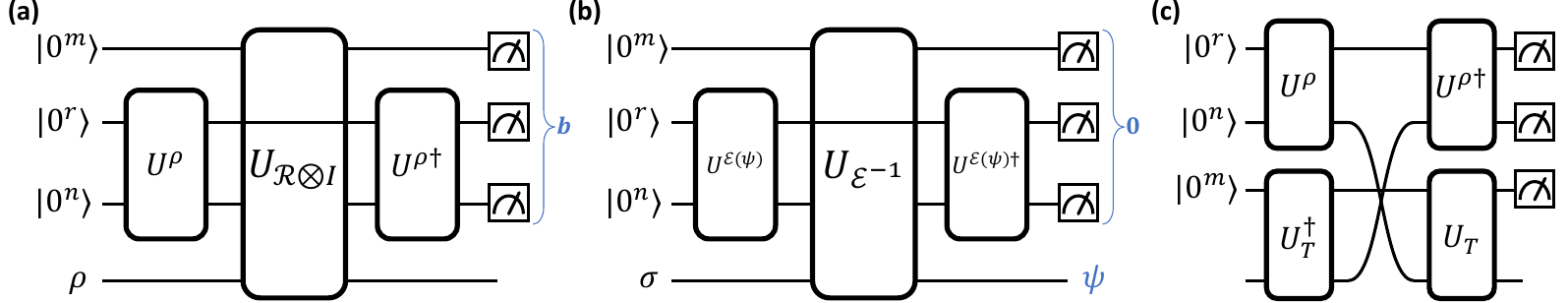}
\caption{(a) The protocol for entanglement detection using the reduction criterion. Measure the first $m+r+n$ qubits and get $\mathbf{b}$. Repeat this circuit for $K$ times, and get $\{\mathbf{b}_1,\cdots,\mathbf{b}_K\}$. If $\mathbf{0}=0^{m+r+n}\in\{\mathbf{b}_1,\cdots,\mathbf{b}_K\}$, we classify $\rho$ as entangled, otherwise we classify $\rho$ as separable. (b) The protocol for inverting a known quantum channel on an unknown state. If the measure result on the first $m+r+n$ qubits is $\mathbf{0}$, then the quantum channel $\mc{E}$ is inverted, and the last $n$ qubit outputs $\psi$. 
The success probability can be enhanced by robust oblivious amplitude amplification.
(c) The circuit for block encoding $T\rho T^\dagger$, where $U_T$ is a block encoding for $T$.}
\label{fig:applications}
\end{figure*}

\subsection{Application in entanglement detection}

Entanglement, the hallmark of non-classical correlations in composite quantum systems, underpins quantum computational supremacy \cite{Harrow2017supremacy,Chitambar2019resource} and a broad range of quantum-information tasks \cite{ekert1991quantum,Bennett1992Communication,Bennett1993Teleporting,Raussendorf2001OneWay}.  
It has been studied systematically for more than three decades \cite{Pasquale2004qft,Amico2008Entanglement,Horodecki2009entanglement,gunhe2009entanglement,pan2012multiphoton}.  
Entanglement criteria \cite{gunhe2009entanglement} not only provide rigorous proof of entanglement but also guide the design of entanglement detection protocols \cite{liu2024separation} to benchmark and validate entanglement experimentally.

Among the many entanglement criteria for bipartite systems, positive map criteria stand out for their conceptual simplicity and strong detection power \cite{gunhe2009entanglement}.  
Given a positive but non-CP map $\mathcal P$, the criterion asserts that a bipartite state $\rho$ on $A\otimes B$ is entangled if $\bigl(\mathcal P_A\otimes\mathcal I_B\bigr)(\rho)$ has a negative eigenvalue.
Prominent examples include the positive partial-transpose criterion \cite{peres1996ppt} (with the positive map as the transpose map) and the reduction criterion \cite{Horodecki1999reduction} (with reduction map $\mathcal R(A)=\Tr(A)\mathbb I-A$). 
Verifying these criteria usually requires highly joint operations or exponential repetition times \cite{gray2018machine,yu2021optimal,zhou2020Single,elben2020mixed}, especially for single-copy protocols \cite{liu2024separation}.

BELT offers a fully coherent route.  
For a bipartite state $\rho$ on $A\otimes B$ and a positive map $\mathcal P$ acting on $A$, we treat the partial positive map $\mathcal N=\mathcal P_A\otimes\mathcal I_B$ as the linear transformation to be simulated.  
BELT yields a block encoding of $\mathcal N(\rho)$, whose smallest (possibly negative) eigenvalue immediately reveals whether $\rho$ is entangled.

As a concrete illustration, we now demonstrate an exponential-versus-constant separation in the sample complexity of a pure state entanglement detection task.

\begin{definition}\label{def:ED_task}
A bipartite state $\rho$ on $n=2q$ qubits is said to be sampled according to the distribution of $\pi$, if with probability $0.5$, it is an $n$-qubit Haar random pure state; with probability $0.5$, it is the tensor product of two independent $q$-qubit Haar random pure states, $\rho=\rho_A\otimes\rho_B$, with $\abs{A}=\abs{B}=q$.
\end{definition}


\begin{fact}[\cite{liu2024separation}]
Given multiple copies of an $n$-qubit state $\rho\sim\pi$, any strategy restricted to single-copy operations requires $\Omega \bigl(2^{n/4}\bigr)$ copies to decide, with success probability at least $2/3$, whether $\rho$ is entangled.
\end{fact}

Notice that the reduction criterion successfully detects entanglement for states sampled from $\pi$.  
If $\ket{\psi}$ is a Haar-random pure state on $2q$ qubits, then
$\mathcal R_A\otimes\mathcal I_B(\psi)=\mathbb I_A\otimes\Tr_A(\psi)-\psi\approx-\psi$,
where “$\approx$’’ holds with probability tending to~$1$ as $n$ grows.  
For a product state $\ket{\psi_A}\otimes\ket{\psi_B}$, we obtain
$\mathcal R_A\otimes\mathcal I_B(\psi_A\otimes\psi_B)=(\mathbb I_A-\psi_A)\otimes\psi_B\ge0$.

Consider the protocol in Fig.~\ref{fig:applications}(a).  
Assume $\rho\sim\pi$ and we have access to multiple copies of $U^\rho$ and its inverse.
The partially transposed Choi matrix of $\mathcal R\otimes\mathcal I$ is
$\Lambda_{\mathcal R\otimes\mathcal I}^{\mathrm T_1}
   =\mathbb I_A\otimes S_B-S_A\otimes S_B$,
where $S_A$ ($S_B$) swaps two copies of $A$ ($B$). We have
$\|\Lambda_{\mathcal R\otimes\mathcal I}^{\mathrm T_1}\|_\infty=2$,
independent of the system size.  
Hence one can choose $U_{\mathcal R\otimes\mathcal I}$ as a $(2,m,0)$ block encoding of $\Lambda_{\mathcal R\otimes\mathcal I}^{\mathrm T_1}$.
We feed $\rho$ into the last system.  
By Theorem~\ref{thm:main}, measuring the first system in the computational basis yields outcome $\ket{0^{m+r+n}}$ with probability $\frac{1}{4}\Tr[\mc{R}\otimes\mc{I}(\rho)\rho\mc{R}\otimes\mc{I}(\rho)]$.
For $\rho=\ketbra{\psi}{\psi}$ Haar-random on $2q$ qubits, $\frac{1}{4}\Tr[\mc{R}\otimes\mc{I}(\rho)\rho\mc{R}\otimes\mc{I}(\rho)]\approx\frac{1}{4}\Tr[(-\psi)\psi(-\psi)]=\frac{1}{4}$,
whereas for $\rho=\psi_A\otimes\psi_B$ with Haar-random $\psi_A$ and $\psi_B$,
$\frac{1}{4}\Tr[\mc{R}\otimes\mc{I}(\rho)\rho\mc{R}\otimes\mc{I}(\rho)]$ equals $0$ because $\Tr[(\mathbb I_A-\psi_A)\psi_A]=0$.  
Repeating the circuit in Fig.~\ref{fig:applications}(a) and declaring “entangled’’ upon observing $\ket{0^{m+r+n}}$ and “separable’’ otherwise gives the following result, proved in Appendix~\ref{app:proof_of_ED}.

\begin{theorem}\label{thm:ED}
The protocol in Fig.~\ref{fig:applications}(a) detects entanglement for $\rho\sim\pi$ using 6 oracle calls to $U^\rho$ or ${U^\rho}^\dagger$, with success probability $ \ge2/3$ for all sufficiently large $n$.
\end{theorem}

The exponential-versus-constant separation in sample complexity between single-copy protocols and the protocol of Fig.~\ref{fig:applications}(a) stems from the enlarged system size ($m+r+n$ ancilla qubits) and the joint operations performed on it.

The gate complexity of the protocol shown in Fig.~\ref{fig:applications}(a) remains efficient: the matrix $\Lambda_{\mathcal R\otimes\mathcal I}^{\mathrm T_1}$ is $2$-sparse and thus admits an efficient block encoding \cite{gilyn2019singular} (see Appendix~\ref{app:sparse_BE}).

\begin{table*}[htbp]
\centering
\resizebox{0.9\textwidth}{!}{
\begin{tabular}{c|c|c|c}
\hline
\hline
 & Map type & Entanglement detection (Def.~\ref{def:ED_task}) & 
Recover $\psi$ from $\mc{E}(\psi)$
 \\
\hline
BELT (Theorem~\ref{thm:main}) & linear & 6 & $\mc{O}\big(\log(\delta^{-1})\bignorm{\Lambda_{\mc{E}^{-1}}^{\mathrm{T}_1}}_{\infty}\bra{\psi}\sigma\ket{\psi}^{-1}\big)$ \\
\hline
HME \cite{wei2023realizing} & Hermitian-preserving & $\mc{O}(1)$ & $\mc{O}\big(\log(\delta^{-1})\epsilon^{-1}\bignorm{\Lambda_{\mc{E}^{-1}}^{\mathrm{T}_1}}_{\infty}^{2}\bra{\psi}\sigma\ket{\psi}^{-2}\big)$ \\
\hline
Single-copy protocols & CPTP & $\Omega(2^{n/4})$ & cannot perform $\mc{E}^{-1}$ \\
\hline
\hline
\end{tabular}
}
\caption{Comparison of the power of BELT, HME, and single-copy protocols.
BELT can simulate any linear map, thus extending the power of HME, which simulates Hermitian-preserving maps.
In the tasks of entanglement detection and inverting quantum channels, compared with HME, BELT offers three key advantages:  
(i) it further decreases the sample complexity for both tasks;  
(ii) it substantially shortens the circuit depth required for entanglement detection; and  
(iii) it removes the $\epsilon$-dependence when reconstructing $\psi$ from $\mathcal{E}(\psi)$.
The sample complexity improvements of BELT compared to HME arise from a stronger oracle model--namely, access to the unitary that prepares a purification of the input state, rather than access to the state itself.
For any single-copy protocol, the sample complexity for the entanglement detection task is exponential in $n$, and the inverse of a CPTP map $\mc{E}$ cannot be performed in general.
}
\label{tab:HME_BE_comparison}
\end{table*}

\subsection{Application in inverting quantum channel}

The task of inverting a quantum channel (i.e., a CPTP map) has received growing attention in recent years \cite{Quintino2019Reversing,Yoshida2023Reversing,Zhu2024Reversing,mo2025efficientinversionunknownunitary,Mo2025Parameterized}.  
Both quantum error correction \cite{shor1995scheme,gottesman1997stabilizer} and quantum error mitigation \cite{cai2022quantum} can be viewed as strategies for (approximately) inverting noise channels, with additional time and space overhead.
If an invertible CPTP map $\mathcal E$ is not unitary, its inverse $\mathcal E^{-1}$ is Hermitian-preserving and trace-preserving (HPTP) but necessarily non-CP. 

Consider the task that, for $\mathcal E$ a known invertible channel and $\psi$ an unknown pure state, try to recover from $\mathcal E(\psi)$.
A common way is to simulate $\mc{E}^{-1}$ ``virtually" \cite{temme2017mitigation,endo2018practical,Jiang2021physical,Regula2021operational,Zhu2024Reversing}: one decomposes the HPTP map $\mc{E}^{-1}$ into an affine combination of CPTP maps and then applies quasi-probability sampling to choose which CPTP map to implement.  
Such virtual methods yield expectation values of the original state $\psi=\mc{E}^{-1}\circ\mc{E}(\psi)$ but cannot directly prepare $\psi$.

Using BELT, we can prepare the original state $\psi$ from $\mc{E}(\psi)$.
Assume we have an oracle $U^{\mathcal E(\psi)}$ that prepares a purification of $\mathcal E(\psi)$, together with its inverse $(U^{\mathcal E(\psi)})^{\dagger}$. Notice that $(U^{\mathcal E(\psi)})^{\dagger}$ can also be obtained from $U^{\mathcal E(\psi)}$ with existing techniques for unitary inversion \cite{Quintino2019Reversing,Yoshida2023Reversing,mo2025efficientinversionunknownunitary,Mo2025Parameterized}.
In the protocol shown in Fig.~\ref{fig:applications}(b), the unitary $U_{\mathcal E^{-1}}$ is chosen as a $(\bignorm{\Lambda_{\mathcal E^{-1}}^{\mathrm{T}_1}}_{\infty},m,0)$-block encoding of $\Lambda_{\mathcal E^{-1}}^{\mathrm{T}_1}$, and $\sigma$ is a reference state that is a guess for $\psi$.  
The overall circuit therefore implements a $(\bignorm{\Lambda_{\mathcal E^{-1}}^{\mathrm{T}_1}}_{\infty},m+r+n,0)$-block encoding of $\mathcal E^{-1}(\mathcal E(\psi))=\psi$.  
Upon measuring the first $m+r+n$ qubits and obtaining $\mathbf 0=0^{m+r+n}$, the remaining system collapses to $\ketbra{\psi}{\psi}$.  
The success probability of this post-selection equals $\bignorm{\Lambda_{\mathcal E^{-1}}^{\mathrm{T}_1}}_{\infty}^{-2}\bra{\psi}\sigma\ket{\psi}$.  
The circuit shown in  Fig.~\ref{fig:applications}(b) is repeated until the measurement result $\mathbf{0}$ occurs.
Therefore, to achieve success probability $\ge1-\delta$, we need to repeat this circuit for $\mc{O}\big(\log(\delta^{-1})\bignorm{\Lambda_{\mc{E}^{-1}}^{\mathrm{T}_1}}_{\infty}^{2}\bra{\psi}\sigma\ket{\psi}^{-1}\big)$ times, and the number of oracle calls to $U^{\mc{E}(\psi)}$ and its inverse is of the same scaling.

Robust oblivious amplitude amplification~\cite{Berry2014Exponential} can be implemented via quantum singular value transformation (Theorem~15 in Ref.~\cite{gilyn2019singular}) to realize the degree-$N$ Chebyshev polynomial $T_N$, where $N=\mc{O}\bigl(\bignorm{\Lambda_{\mc{E}^{-1}}^{\mathrm{T}_1}}_{\infty}\bigr)$ is an integer. This amplifies the block-encoded $\bignorm{\Lambda_{\mc{E}^{-1}}^{\mathrm{T}_1}}_{\infty}^{-1}\ketbra{\psi}{\psi}$ to $\Omega(1)\ketbra{\psi}{\psi}$.
The procedure requires $N$ applications of $(\mathbb I_m\otimes(U^{\mathcal E(\psi)})^{\dagger}\otimes\mathbb I_n)(U_{\mathcal E^{-1}}\otimes\mathbb I_r)(\mathbb I_m\otimes U^{\mathcal E(\psi)}\otimes\mathbb I_n)$
and its inverse.
The probability of obtaining the all-zero outcome on the $(1+m+r+n)$-qubit ancilla register (the extra qubit comes from the QSVT circuit) is then $\Omega(1)\bra{\psi}\sigma\ket{\psi}$.
Therefore, achieving success probability at least $1-\delta$ in recovering $\psi$ uses $\mathcal O\bigl(\log(\delta^{-1})\bra{\psi}\sigma\ket{\psi}^{-1}\bignorm{\Lambda_{\mathcal E^{-1}}^{\mathrm{T}_1}}_{\infty}\bigr)$ oracle calls to $U^{\mc{E}(\psi)}$ and its inverse.
This protocol couples BELT with QSVT, realizing the transform $T_N \circ\bignorm{\Lambda_{\mc{E}^{-1}}^{\mathrm{T}_1}}_{\infty}^{-1}\mc{E}^{-1}$
on $\mc{E}(\psi)$.

\begin{theorem}\label{thm:inverting}
The protocol shown in Fig.~\ref{fig:applications}(b) can recover $\psi$, with success probability $\ge1-\delta$, using 
\begin{equation}
\order{\log(\delta^{-1})\bignorm{\Lambda_{\mc{E}^{-1}}^{\mathrm{T}_1}}_{\infty}^{2}\bra{\psi}\sigma\ket{\psi}^{-1}}
\end{equation}
calls of $U^{\mc{E}(\psi)}$ and its inverse.
If the robust oblivious amplitude amplification is applied, then the number of calls can be reduced to
\begin{equation}
\order{\log(\delta^{-1})\bignorm{\Lambda_{\mc{E}^{-1}}^{\mathrm{T}_1}}_{\infty}\bra{\psi}\sigma\ket{\psi}^{-1}}.
\end{equation}
\end{theorem}

Compared with noiseless state recovery via HME \cite{wei2023realizing}, which achieves accuracy $\epsilon$ and success probability at least $1-\delta$ with
$\mathcal O\bigl(\log(\delta^{-1})\epsilon^{-1}
\bignorm{\Lambda_{\mathcal E^{-1}}^{\mathrm{T}_1}}_{\infty}^{2}
\bra{\psi}\sigma\ket{\psi}^{-2}\bigr)$
copies of $\mathcal E(\psi)$, the protocol in Fig.~\ref{fig:applications}(b) recovers $\psi$ exactly (no $\epsilon$ dependence), improves the dependence on $\bra{\psi}\sigma\ket{\psi}^{-1}$, and uses a shallower circuit.  
Applying robust oblivious amplitude amplification increases circuit depth but further improves the dependence on $\bignorm{\Lambda_{\mathcal E^{-1}}^{\mathrm{T}_1}}_{\infty}$.  
The improvement stems from having oracle access to a unitary that prepares a purification of $\mc{E}(\psi)$, rather than relying on multiple copies of the mixed state $\mc{E}(\psi)$ itself, as in HME-based noiseless state recovery.

\subsection{Application in pseudo-differential operators}
Although BELT was introduced primarily for simulating non-CP linear maps, it can also handle CP maps. 
By Stinespring’s dilation theorem \cite{Stinespring1955Positive,wood2011tensor}, any CP map $\mc{F}:L(\mc{X})\rightarrow L(\mc{Y})$ can be written as $\mc{F}(\rho)=\Tr_{\mc{Z}}(A\rho A^\dagger)$, where $A:\mc{X}\rightarrow\mc{Y}\otimes\mc{Z}$ is a matrix and $\dim(\mc{Z})\le\dim(\mc{X})\dim(\mc{Y})$.  We have
\begin{align}
\Lambda_{\mc{F}}^{\mathrm{T}_1}=&(\mbb{I}_{\mc{X}}\otimes\mbb{I}_{\mc{Y}}\otimes\bra{\Phi^+})(\mbb{I}_\mc{X}\otimes A\otimes\mbb{I}_\mc{Z})\times\\&(S\otimes\mbb{I}_{\mc{Z}})(\mbb{I}_\mc{X}\otimes A^\dagger\otimes\mbb{I}_\mc{Z})(\mbb{I}_{\mc{X}}\otimes\mbb{I}_{\mc{Y}}\otimes\ket{\Phi^+}),
\end{align}
where $S$ is the SWAP operator acting on $\mc{X}\otimes\mc{X}$ and $\ket{\Phi^+}$ is the unnormalised maximally entangled state on $\mc{Z}\otimes\mc{Z}$.  Note that $\bignorm{\Lambda_{\mc{F}}^{\mathrm{T}_1}}_\infty\le\norm{A}_\infty^{2}$, so BELT enables an efficient block encoding of $\mc{F}(\rho)$ whenever $\norm{A}_\infty$ is bounded.

A special case is $\dim(\mc{Z})=1$, in which case the CP map $\mc{F}(\rho)=A\rho A^\dagger$ represents a possibly non-unitary evolution.
Here a block encoding of $\Lambda_{\mc{F}}^{\mathrm{T}_1}=(\mbb{I}_{\mc{X}}\otimes A)S(\mbb{I}_{\mc{X}}\otimes A^\dagger)$ can be obtained if the block encoding of $A$ is available.

Pseudo-differential operators (PDOs) extend differential operators and arise widely in applied mathematics, including integral-differential equations, differential geometry, and quantum field theory. A PDO takes the form
\begin{equation}
    P(x,D) f(x) = \int_{\mathbb{R}^d} e^{2\pi i x\xi} a(x,\xi) \hat{f}(\xi) d \xi,
\end{equation}
where $a(x,\xi)$ is the symbol of $P(x,D)$ and $\hat{f}$ is the Fourier transform of $f$.  
For the elliptic operator $P(x,D)=I-\nabla\cdot(\omega(x)\nabla)$, the symbol is $a(x,\xi)=1-2\pi i\nabla\omega(x)\cdot\xi+4\pi^{2}|\xi|^{2}$.  
Hence the solution to $P(x,D)u(x)=f(x)$ can be represented as
\begin{equation}
u(x)=P^{-1}(x,D)f(x)=\int_{\mathbb{R}^d}e^{2\pi i x\xi}a^{-1}(x,\xi)\hat{f}(\xi)d\xi.
\end{equation}
The discretized PDO~\cite{li2023efficient} is given by:
\begin{equation}\label{eq:discretized_PDO}
T(x,D) f(x) = \sum_{\xi\in [P-1]_0^d} e^{2\pi i x\xi/P} a(x,\xi) \hat{f}(\xi), \quad x\in[P-1]_0^d,
\end{equation}
where $[P-1]_0 := \{0,1,\ldots,P-1\}$ and $P=2^p$ is the grid number on each coordinate.



Consider the linear map $\rho\mapsto T\rho T^\dagger$, where $T=T(x,D)$ from Eq.~\eqref{eq:discretized_PDO}.  
$\rho$ can describe the input function (for example, the inhomogeneity of an elliptic equation) and $T\rho T^\dagger$ represents the solution.  
Because $T$ admits a block encoding with gate complexity $O\bigl(\operatorname{poly}(pd)+\operatorname{polylog}(1/\epsilon)\bigr)$ \cite[Theorem 5]{li2023efficient}, we can efficiently block encode the solution $T\rho T^\dagger$ using the circuit shown in Fig.~\ref{fig:applications}(c), thereby enabling flexible further modifications.

\section{Discussion}\label{sec:conclusion}

To circumvent the non-physicality of non-CP maps, we propose BELT, a quantum algorithmic primitive for applying linear transformations $\mc{N}$ to a block-encoded density matrix $\rho$, obtained by generalizing the technique for block encoding a density matrix~\cite{Low2019hamiltonian,gilyn2019singular}.
BELT enables the simulation of transformations that are inaccessible to QSVT.
Applications of BELT include entanglement detection, inversion of quantum channels, and the implementation of pseudo-differential operators.

We also propose, in Appendix~\ref{section:BlockEncoding_by_Exponentiation}, an alternative approach that approximately block encodes $\mc{N}(\rho)$ when $\mc{N}$ is Hermitian-preserving, by combining QETU~\cite{dong2022ground} with HME~\cite{wei2023realizing}.

Several avenues for further investigation remain. It would be valuable to identify additional applications of BELT and to establish lower bounds on its sample complexity. For instance, BELT can be applied to simulate Lindbladian systems and stochastic differential equations. Moreover, because the sample complexity in both BELT and HME is controlled by $\bignorm{\Lambda_{\mc{N}}^{\mathrm{T}_1}}_\infty$, elucidating the operational meaning of this quantity in general or specific cases is of great importance.

\begin{acknowledgments}
We appreciate insightful discussions with Tongyang Li, Zhenhuan Liu and Yukun Zhang.
JPL acknowledges support from Innovation Program for Quantum Science and Technology (Grant No.2024ZD0300502), start-up funding from Tsinghua University and Beijing Institute of Mathematical Sciences and Applications.
ZL was supported by Beijing Natural Science Foundation Key Program (Grant No. Z220002) and by NKPs (Grant
no. 2020YFA0713000). FW, RL, YS, JL, and ZL were supported by BMSTC and ACZSP (Grant No. Z221100002722017).
\end{acknowledgments}


\begin{thebibliography}{68}%
	\makeatletter
	\providecommand \@ifxundefined [1]{%
		\@ifx{#1\undefined}
	}%
	\providecommand \@ifnum [1]{%
		\ifnum #1\expandafter \@firstoftwo
		\else \expandafter \@secondoftwo
		\fi
	}%
	\providecommand \@ifx [1]{%
		\ifx #1\expandafter \@firstoftwo
		\else \expandafter \@secondoftwo
		\fi
	}%
	\providecommand \natexlab [1]{#1}%
	\providecommand \enquote  [1]{``#1''}%
	\providecommand \bibnamefont  [1]{#1}%
	\providecommand \bibfnamefont [1]{#1}%
	\providecommand \citenamefont [1]{#1}%
	\providecommand \href@noop [0]{\@secondoftwo}%
	\providecommand \href [0]{\begingroup \@sanitize@url \@href}%
	\providecommand \@href[1]{\@@startlink{#1}\@@href}%
	\providecommand \@@href[1]{\endgroup#1\@@endlink}%
	\providecommand \@sanitize@url [0]{\catcode `\\12\catcode `\$12\catcode
		`\&12\catcode `\#12\catcode `\^12\catcode `\_12\catcode `\%12\relax}%
	\providecommand \@@startlink[1]{}%
	\providecommand \@@endlink[0]{}%
	\providecommand \url  [0]{\begingroup\@sanitize@url \@url }%
	\providecommand \@url [1]{\endgroup\@href {#1}{\urlprefix }}%
	\providecommand \urlprefix  [0]{URL }%
	\providecommand \Eprint [0]{\href }%
	\providecommand \doibase [0]{https://doi.org/}%
	\providecommand \selectlanguage [0]{\@gobble}%
	\providecommand \bibinfo  [0]{\@secondoftwo}%
	\providecommand \bibfield  [0]{\@secondoftwo}%
	\providecommand \translation [1]{[#1]}%
	\providecommand \BibitemOpen [0]{}%
	\providecommand \bibitemStop [0]{}%
	\providecommand \bibitemNoStop [0]{.\EOS\space}%
	\providecommand \EOS [0]{\spacefactor3000\relax}%
	\providecommand \BibitemShut  [1]{\csname bibitem#1\endcsname}%
	\let\auto@bib@innerbib\@empty
	\bibitem [{\citenamefont {Nielsen}\ and\ \citenamefont
		{Chuang}(2010)}]{nielsen2010quantum}%
	\BibitemOpen
	\bibfield  {author} {\bibinfo {author} {\bibfnamefont {M.~A.}\ \bibnamefont
			{Nielsen}}\ and\ \bibinfo {author} {\bibfnamefont {I.~L.}\ \bibnamefont
			{Chuang}},\ }\href {https://doi.org/10.1017/CBO9780511976667} {\emph
		{\bibinfo {title} {Quantum Computation and Quantum Information: 10th
				Anniversary Edition}}}\ (\bibinfo  {publisher} {Cambridge University Press},\
	\bibinfo {year} {2010})\BibitemShut {NoStop}%
	\bibitem [{\citenamefont {Watrous}(2018)}]{watrous2018theory}%
	\BibitemOpen
	\bibfield  {author} {\bibinfo {author} {\bibfnamefont {J.}~\bibnamefont
			{Watrous}},\ }\href {https://doi.org/10.1017/9781316848142} {\emph {\bibinfo
			{title} {The Theory of Quantum Information}}}\ (\bibinfo  {publisher}
	{Cambridge University Press},\ \bibinfo {address} {Cambridge},\ \bibinfo
	{year} {2018})\BibitemShut {NoStop}%
	\bibitem [{\citenamefont {Pechukas}(1994)}]{Pechukas1994reduced}%
	\BibitemOpen
	\bibfield  {author} {\bibinfo {author} {\bibfnamefont {P.}~\bibnamefont
			{Pechukas}},\ }\bibfield  {title} {\bibinfo {title} {Reduced dynamics need
			not be completely positive},\ }\href
	{https://doi.org/10.1103/PhysRevLett.73.1060} {\bibfield  {journal} {\bibinfo
			{journal} {Phys. Rev. Lett.}\ }\textbf {\bibinfo {volume} {73}},\ \bibinfo
		{pages} {1060} (\bibinfo {year} {1994})}\BibitemShut {NoStop}%
	\bibitem [{\citenamefont {Gühne}\ and\ \citenamefont
		{Tóth}(2009)}]{gunhe2009entanglement}%
	\BibitemOpen
	\bibfield  {author} {\bibinfo {author} {\bibfnamefont {O.}~\bibnamefont
			{Gühne}}\ and\ \bibinfo {author} {\bibfnamefont {G.}~\bibnamefont {Tóth}},\
	}\bibfield  {title} {\bibinfo {title} {Entanglement detection},\ }\href
	{https://doi.org/https://doi.org/10.1016/j.physrep.2009.02.004} {\bibfield
		{journal} {\bibinfo  {journal} {Physics Reports}\ }\textbf {\bibinfo {volume}
			{474}},\ \bibinfo {pages} {1} (\bibinfo {year} {2009})}\BibitemShut {NoStop}%
	\bibitem [{\citenamefont {Cai}\ \emph {et~al.}(2022)\citenamefont {Cai},
		\citenamefont {Babbush}, \citenamefont {Benjamin}, \citenamefont {Endo},
		\citenamefont {Huggins}, \citenamefont {Li}, \citenamefont {McClean},\ and\
		\citenamefont {O'Brien}}]{cai2022quantum}%
	\BibitemOpen
	\bibfield  {author} {\bibinfo {author} {\bibfnamefont {Z.}~\bibnamefont
			{Cai}}, \bibinfo {author} {\bibfnamefont {R.}~\bibnamefont {Babbush}},
		\bibinfo {author} {\bibfnamefont {S.~C.}\ \bibnamefont {Benjamin}}, \bibinfo
		{author} {\bibfnamefont {S.}~\bibnamefont {Endo}}, \bibinfo {author}
		{\bibfnamefont {W.~J.}\ \bibnamefont {Huggins}}, \bibinfo {author}
		{\bibfnamefont {Y.}~\bibnamefont {Li}}, \bibinfo {author} {\bibfnamefont
			{J.~R.}\ \bibnamefont {McClean}},\ and\ \bibinfo {author} {\bibfnamefont
			{T.~E.}\ \bibnamefont {O'Brien}},\ }\bibfield  {title} {\bibinfo {title}
		{Quantum error mitigation},\ }\href {https://arxiv.org/abs/2210.00921}
	{\bibfield  {journal} {\bibinfo  {journal} {arXiv:2210.00921}\ } (\bibinfo
		{year} {2022})}\BibitemShut {NoStop}%
	\bibitem [{\citenamefont {Endo}\ \emph {et~al.}(2021)\citenamefont {Endo},
		\citenamefont {Cai}, \citenamefont {Benjamin},\ and\ \citenamefont
		{Yuan}}]{endo2021hybrid}%
	\BibitemOpen
	\bibfield  {author} {\bibinfo {author} {\bibfnamefont {S.}~\bibnamefont
			{Endo}}, \bibinfo {author} {\bibfnamefont {Z.}~\bibnamefont {Cai}}, \bibinfo
		{author} {\bibfnamefont {S.~C.}\ \bibnamefont {Benjamin}},\ and\ \bibinfo
		{author} {\bibfnamefont {X.}~\bibnamefont {Yuan}},\ }\bibfield  {title}
	{\bibinfo {title} {Hybrid quantum-classical algorithms and quantum error
			mitigation},\ }\href {https://doi.org/10.7566/JPSJ.90.032001} {\bibfield
		{journal} {\bibinfo  {journal} {Journal of the Physical Society of Japan}\
		}\textbf {\bibinfo {volume} {90}},\ \bibinfo {pages} {032001} (\bibinfo
		{year} {2021})}\BibitemShut {NoStop}%
	\bibitem [{\citenamefont {Wei}\ \emph {et~al.}(2024)\citenamefont {Wei},
		\citenamefont {Liu}, \citenamefont {Liu}, \citenamefont {Han}, \citenamefont
		{Deng},\ and\ \citenamefont {Liu}}]{wei2023realizing}%
	\BibitemOpen
	\bibfield  {author} {\bibinfo {author} {\bibfnamefont {F.}~\bibnamefont
			{Wei}}, \bibinfo {author} {\bibfnamefont {Z.}~\bibnamefont {Liu}}, \bibinfo
		{author} {\bibfnamefont {G.}~\bibnamefont {Liu}}, \bibinfo {author}
		{\bibfnamefont {Z.}~\bibnamefont {Han}}, \bibinfo {author} {\bibfnamefont
			{D.-L.}\ \bibnamefont {Deng}},\ and\ \bibinfo {author} {\bibfnamefont
			{Z.}~\bibnamefont {Liu}},\ }\bibfield  {title} {\bibinfo {title} {Simulating
			non-completely positive actions via exponentiation of hermitian-preserving
			maps},\ }\href {https://doi.org/10.1038/s41534-024-00949-z} {\bibfield
		{journal} {\bibinfo  {journal} {npj Quantum Information}\ }\textbf {\bibinfo
			{volume} {10}},\ \bibinfo {pages} {134} (\bibinfo {year} {2024})}\BibitemShut
	{NoStop}%
	\bibitem [{\citenamefont {Son}\ \emph {et~al.}(2025)\citenamefont {Son},
		\citenamefont {Gluza}, \citenamefont {Takagi},\ and\ \citenamefont
		{Ng}}]{Son2025Dynamic}%
	\BibitemOpen
	\bibfield  {author} {\bibinfo {author} {\bibfnamefont {J.}~\bibnamefont
			{Son}}, \bibinfo {author} {\bibfnamefont {M.}~\bibnamefont {Gluza}}, \bibinfo
		{author} {\bibfnamefont {R.}~\bibnamefont {Takagi}},\ and\ \bibinfo {author}
		{\bibfnamefont {N.~H.~Y.}\ \bibnamefont {Ng}},\ }\bibfield  {title} {\bibinfo
		{title} {Quantum dynamic programming},\ }\href
	{https://doi.org/10.1103/PhysRevLett.134.180602} {\bibfield  {journal}
		{\bibinfo  {journal} {Phys. Rev. Lett.}\ }\textbf {\bibinfo {volume} {134}},\
		\bibinfo {pages} {180602} (\bibinfo {year} {2025})}\BibitemShut {NoStop}%
	\bibitem [{\citenamefont {Salgado}\ \emph {et~al.}(2004)\citenamefont
		{Salgado}, \citenamefont {S\'anchez-G\'omez},\ and\ \citenamefont
		{Ferrero}}]{Salgado2004evolution}%
	\BibitemOpen
	\bibfield  {author} {\bibinfo {author} {\bibfnamefont {D.}~\bibnamefont
			{Salgado}}, \bibinfo {author} {\bibfnamefont {J.~L.}\ \bibnamefont
			{S\'anchez-G\'omez}},\ and\ \bibinfo {author} {\bibfnamefont
			{M.}~\bibnamefont {Ferrero}},\ }\bibfield  {title} {\bibinfo {title}
		{Evolution of any finite open quantum system always admits a kraus-type
			representation, although it is not always completely positive},\ }\href
	{https://doi.org/10.1103/PhysRevA.70.054102} {\bibfield  {journal} {\bibinfo
			{journal} {Phys. Rev. A}\ }\textbf {\bibinfo {volume} {70}},\ \bibinfo
		{pages} {054102} (\bibinfo {year} {2004})}\BibitemShut {NoStop}%
	\bibitem [{\citenamefont {Carteret}\ \emph {et~al.}(2008)\citenamefont
		{Carteret}, \citenamefont {Terno},\ and\ \citenamefont {\ifmmode~\dot{Z}\else
			\.{Z}\fi{}yczkowski}}]{carteret2008dynamics}%
	\BibitemOpen
	\bibfield  {author} {\bibinfo {author} {\bibfnamefont {H.~A.}\ \bibnamefont
			{Carteret}}, \bibinfo {author} {\bibfnamefont {D.~R.}\ \bibnamefont
			{Terno}},\ and\ \bibinfo {author} {\bibfnamefont {K.}~\bibnamefont
			{\ifmmode~\dot{Z}\else \.{Z}\fi{}yczkowski}},\ }\bibfield  {title} {\bibinfo
		{title} {Dynamics beyond completely positive maps: Some properties and
			applications},\ }\href {https://doi.org/10.1103/PhysRevA.77.042113}
	{\bibfield  {journal} {\bibinfo  {journal} {Phys. Rev. A}\ }\textbf {\bibinfo
			{volume} {77}},\ \bibinfo {pages} {042113} (\bibinfo {year}
		{2008})}\BibitemShut {NoStop}%
	\bibitem [{\citenamefont {Temme}\ \emph {et~al.}(2017)\citenamefont {Temme},
		\citenamefont {Bravyi},\ and\ \citenamefont
		{Gambetta}}]{temme2017mitigation}%
	\BibitemOpen
	\bibfield  {author} {\bibinfo {author} {\bibfnamefont {K.}~\bibnamefont
			{Temme}}, \bibinfo {author} {\bibfnamefont {S.}~\bibnamefont {Bravyi}},\ and\
		\bibinfo {author} {\bibfnamefont {J.~M.}\ \bibnamefont {Gambetta}},\
	}\bibfield  {title} {\bibinfo {title} {Error mitigation for short-depth
			quantum circuits},\ }\href {https://doi.org/10.1103/PhysRevLett.119.180509}
	{\bibfield  {journal} {\bibinfo  {journal} {Phys. Rev. Lett.}\ }\textbf
		{\bibinfo {volume} {119}},\ \bibinfo {pages} {180509} (\bibinfo {year}
		{2017})}\BibitemShut {NoStop}%
	\bibitem [{\citenamefont {Endo}\ \emph {et~al.}(2018)\citenamefont {Endo},
		\citenamefont {Benjamin},\ and\ \citenamefont {Li}}]{endo2018practical}%
	\BibitemOpen
	\bibfield  {author} {\bibinfo {author} {\bibfnamefont {S.}~\bibnamefont
			{Endo}}, \bibinfo {author} {\bibfnamefont {S.~C.}\ \bibnamefont {Benjamin}},\
		and\ \bibinfo {author} {\bibfnamefont {Y.}~\bibnamefont {Li}},\ }\bibfield
	{title} {\bibinfo {title} {Practical quantum error mitigation for near-future
			applications},\ }\href {https://doi.org/10.1103/PhysRevX.8.031027} {\bibfield
		{journal} {\bibinfo  {journal} {Phys. Rev. X}\ }\textbf {\bibinfo {volume}
			{8}},\ \bibinfo {pages} {031027} (\bibinfo {year} {2018})}\BibitemShut
	{NoStop}%
	\bibitem [{\citenamefont {Zhao}\ \emph {et~al.}(2025)\citenamefont {Zhao},
		\citenamefont {Zhang}, \citenamefont {Zhao},\ and\ \citenamefont
		{Wang}}]{Zhao2025Power}%
	\BibitemOpen
	\bibfield  {author} {\bibinfo {author} {\bibfnamefont {X.}~\bibnamefont
			{Zhao}}, \bibinfo {author} {\bibfnamefont {L.}~\bibnamefont {Zhang}},
		\bibinfo {author} {\bibfnamefont {B.}~\bibnamefont {Zhao}},\ and\ \bibinfo
		{author} {\bibfnamefont {X.}~\bibnamefont {Wang}},\ }\bibfield  {title}
	{\bibinfo {title} {Power of quantum measurement in simulating unphysical
			operations},\ }\href {https://doi.org/10.1103/PhysRevResearch.7.013334}
	{\bibfield  {journal} {\bibinfo  {journal} {Phys. Rev. Res.}\ }\textbf
		{\bibinfo {volume} {7}},\ \bibinfo {pages} {013334} (\bibinfo {year}
		{2025})}\BibitemShut {NoStop}%
	\bibitem [{\citenamefont {Horodecki}\ and\ \citenamefont
		{Ekert}(2002)}]{horodecki2002direct}%
	\BibitemOpen
	\bibfield  {author} {\bibinfo {author} {\bibfnamefont {P.}~\bibnamefont
			{Horodecki}}\ and\ \bibinfo {author} {\bibfnamefont {A.}~\bibnamefont
			{Ekert}},\ }\bibfield  {title} {\bibinfo {title} {Method for direct detection
			of quantum entanglement},\ }\href
	{https://doi.org/10.1103/PhysRevLett.89.127902} {\bibfield  {journal}
		{\bibinfo  {journal} {Phys. Rev. Lett.}\ }\textbf {\bibinfo {volume} {89}},\
		\bibinfo {pages} {127902} (\bibinfo {year} {2002})}\BibitemShut {NoStop}%
	\bibitem [{\citenamefont {Korbicz}\ \emph {et~al.}(2008)\citenamefont
		{Korbicz}, \citenamefont {Almeida}, \citenamefont {Bae}, \citenamefont
		{Lewenstein},\ and\ \citenamefont {Ac\'{\i}n}}]{korb2008structural}%
	\BibitemOpen
	\bibfield  {author} {\bibinfo {author} {\bibfnamefont {J.~K.}\ \bibnamefont
			{Korbicz}}, \bibinfo {author} {\bibfnamefont {M.~L.}\ \bibnamefont
			{Almeida}}, \bibinfo {author} {\bibfnamefont {J.}~\bibnamefont {Bae}},
		\bibinfo {author} {\bibfnamefont {M.}~\bibnamefont {Lewenstein}},\ and\
		\bibinfo {author} {\bibfnamefont {A.}~\bibnamefont {Ac\'{\i}n}},\ }\bibfield
	{title} {\bibinfo {title} {Structural approximations to positive maps and
			entanglement-breaking channels},\ }\href
	{https://doi.org/10.1103/PhysRevA.78.062105} {\bibfield  {journal} {\bibinfo
			{journal} {Phys. Rev. A}\ }\textbf {\bibinfo {volume} {78}},\ \bibinfo
		{pages} {062105} (\bibinfo {year} {2008})}\BibitemShut {NoStop}%
	\bibitem [{\citenamefont {Dong}\ \emph {et~al.}(2019)\citenamefont {Dong},
		\citenamefont {Quintino}, \citenamefont {Soeda},\ and\ \citenamefont
		{Murao}}]{dong2019positive}%
	\BibitemOpen
	\bibfield  {author} {\bibinfo {author} {\bibfnamefont {Q.}~\bibnamefont
			{Dong}}, \bibinfo {author} {\bibfnamefont {M.~T.}\ \bibnamefont {Quintino}},
		\bibinfo {author} {\bibfnamefont {A.}~\bibnamefont {Soeda}},\ and\ \bibinfo
		{author} {\bibfnamefont {M.}~\bibnamefont {Murao}},\ }\bibfield  {title}
	{\bibinfo {title} {Implementing positive maps with multiple copies of an
			input state},\ }\href {https://doi.org/10.1103/PhysRevA.99.052352} {\bibfield
		{journal} {\bibinfo  {journal} {Phys. Rev. A}\ }\textbf {\bibinfo {volume}
			{99}},\ \bibinfo {pages} {052352} (\bibinfo {year} {2019})}\BibitemShut
	{NoStop}%
	\bibitem [{\citenamefont {Harrow}\ \emph {et~al.}(2009)\citenamefont {Harrow},
		\citenamefont {Hassidim},\ and\ \citenamefont {Lloyd}}]{Harrow2009HHL}%
	\BibitemOpen
	\bibfield  {author} {\bibinfo {author} {\bibfnamefont {A.~W.}\ \bibnamefont
			{Harrow}}, \bibinfo {author} {\bibfnamefont {A.}~\bibnamefont {Hassidim}},\
		and\ \bibinfo {author} {\bibfnamefont {S.}~\bibnamefont {Lloyd}},\ }\bibfield
	{title} {\bibinfo {title} {Quantum algorithm for linear systems of
			equations},\ }\href {https://doi.org/10.1103/PhysRevLett.103.150502}
	{\bibfield  {journal} {\bibinfo  {journal} {Phys. Rev. Lett.}\ }\textbf
		{\bibinfo {volume} {103}},\ \bibinfo {pages} {150502} (\bibinfo {year}
		{2009})}\BibitemShut {NoStop}%
	\bibitem [{\citenamefont {Gily\'{e}n}\ \emph {et~al.}(2019)\citenamefont
		{Gily\'{e}n}, \citenamefont {Su}, \citenamefont {Low},\ and\ \citenamefont
		{Wiebe}}]{gilyn2019singular}%
	\BibitemOpen
	\bibfield  {author} {\bibinfo {author} {\bibfnamefont {A.}~\bibnamefont
			{Gily\'{e}n}}, \bibinfo {author} {\bibfnamefont {Y.}~\bibnamefont {Su}},
		\bibinfo {author} {\bibfnamefont {G.~H.}\ \bibnamefont {Low}},\ and\ \bibinfo
		{author} {\bibfnamefont {N.}~\bibnamefont {Wiebe}},\ }\bibfield  {title}
	{\bibinfo {title} {Quantum singular value transformation and beyond:
			Exponential improvements for quantum matrix arithmetics},\ }in\ \href
	{https://doi.org/10.1145/3313276.3316366} {\emph {\bibinfo {booktitle}
			{Proceedings of the 51st Annual ACM SIGACT Symposium on Theory of
				Computing}}},\ \bibinfo {series and number} {STOC 2019}\ (\bibinfo
	{publisher} {Association for Computing Machinery},\ \bibinfo {address} {New
		York, NY, USA},\ \bibinfo {year} {2019})\ p.\ \bibinfo {pages}
	{193–204}\BibitemShut {NoStop}%
	\bibitem [{\citenamefont {Dong}\ and\ \citenamefont {Lin}(2021)}]{Dong_2021}%
	\BibitemOpen
	\bibfield  {author} {\bibinfo {author} {\bibfnamefont {Y.}~\bibnamefont
			{Dong}}\ and\ \bibinfo {author} {\bibfnamefont {L.}~\bibnamefont {Lin}},\
	}\bibfield  {title} {\bibinfo {title} {Random circuit block-encoded matrix
			and a proposal of quantum linpack benchmark},\ }\bibfield  {journal}
	{\bibinfo  {journal} {Physical Review A}\ }\textbf {\bibinfo {volume}
		{103}},\ \href {https://doi.org/10.1103/physreva.103.062412}
	{10.1103/physreva.103.062412} (\bibinfo {year} {2021})\BibitemShut {NoStop}%
	\bibitem [{\citenamefont {Costa}\ \emph {et~al.}(2022)\citenamefont {Costa},
		\citenamefont {An}, \citenamefont {Sanders}, \citenamefont {Su},
		\citenamefont {Babbush},\ and\ \citenamefont {Berry}}]{Costa2022optimal}%
	\BibitemOpen
	\bibfield  {author} {\bibinfo {author} {\bibfnamefont {P.~C.}\ \bibnamefont
			{Costa}}, \bibinfo {author} {\bibfnamefont {D.}~\bibnamefont {An}}, \bibinfo
		{author} {\bibfnamefont {Y.~R.}\ \bibnamefont {Sanders}}, \bibinfo {author}
		{\bibfnamefont {Y.}~\bibnamefont {Su}}, \bibinfo {author} {\bibfnamefont
			{R.}~\bibnamefont {Babbush}},\ and\ \bibinfo {author} {\bibfnamefont {D.~W.}\
			\bibnamefont {Berry}},\ }\bibfield  {title} {\bibinfo {title} {Optimal
			scaling quantum linear-systems solver via discrete adiabatic theorem},\
	}\href {https://doi.org/10.1103/PRXQuantum.3.040303} {\bibfield  {journal}
		{\bibinfo  {journal} {PRX Quantum}\ }\textbf {\bibinfo {volume} {3}},\
		\bibinfo {pages} {040303} (\bibinfo {year} {2022})}\BibitemShut {NoStop}%
	\bibitem [{\citenamefont {Gily\'en}\ \emph {et~al.}(2022)\citenamefont
		{Gily\'en}, \citenamefont {Lloyd}, \citenamefont {Marvian}, \citenamefont
		{Quek},\ and\ \citenamefont {Wilde}}]{Gilyen2022Petz}%
	\BibitemOpen
	\bibfield  {author} {\bibinfo {author} {\bibfnamefont {A.}~\bibnamefont
			{Gily\'en}}, \bibinfo {author} {\bibfnamefont {S.}~\bibnamefont {Lloyd}},
		\bibinfo {author} {\bibfnamefont {I.}~\bibnamefont {Marvian}}, \bibinfo
		{author} {\bibfnamefont {Y.}~\bibnamefont {Quek}},\ and\ \bibinfo {author}
		{\bibfnamefont {M.~M.}\ \bibnamefont {Wilde}},\ }\bibfield  {title} {\bibinfo
		{title} {Quantum algorithm for {Petz} recovery channels and pretty good
			measurements},\ }\href {https://doi.org/10.1103/PhysRevLett.128.220502}
	{\bibfield  {journal} {\bibinfo  {journal} {Phys. Rev. Lett.}\ }\textbf
		{\bibinfo {volume} {128}},\ \bibinfo {pages} {220502} (\bibinfo {year}
		{2022})}\BibitemShut {NoStop}%
	\bibitem [{\citenamefont {Low}\ and\ \citenamefont
		{Su}(2024)}]{Low2024Eigenvalue}%
	\BibitemOpen
	\bibfield  {author} {\bibinfo {author} {\bibfnamefont {G.~H.}\ \bibnamefont
			{Low}}\ and\ \bibinfo {author} {\bibfnamefont {Y.}~\bibnamefont {Su}},\
	}\bibfield  {title} {\bibinfo {title} {Quantum eigenvalue processing},\ }in\
	\href {https://doi.org/10.1109/FOCS61266.2024.00070} {\emph {\bibinfo
			{booktitle} {2024 IEEE 65th Annual Symposium on Foundations of Computer
				Science (FOCS)}}}\ (\bibinfo {year} {2024})\ pp.\ \bibinfo {pages}
	{1051--1062}\BibitemShut {NoStop}%
	\bibitem [{\citenamefont {Chakraborty}\ \emph {et~al.}(2025)\citenamefont
		{Chakraborty}, \citenamefont {Hazra}, \citenamefont {Li}, \citenamefont
		{Shao}, \citenamefont {Wang},\ and\ \citenamefont
		{Zhang}}]{chakraborty2025quantumsingularvaluetransformation}%
	\BibitemOpen
	\bibfield  {author} {\bibinfo {author} {\bibfnamefont {S.}~\bibnamefont
			{Chakraborty}}, \bibinfo {author} {\bibfnamefont {S.}~\bibnamefont {Hazra}},
		\bibinfo {author} {\bibfnamefont {T.}~\bibnamefont {Li}}, \bibinfo {author}
		{\bibfnamefont {C.}~\bibnamefont {Shao}}, \bibinfo {author} {\bibfnamefont
			{X.}~\bibnamefont {Wang}},\ and\ \bibinfo {author} {\bibfnamefont
			{Y.}~\bibnamefont {Zhang}},\ }\href@noop {} {\bibinfo {title} {Quantum
			singular value transformation without block encodings: Near-optimal
			complexity with minimal ancilla}} (\bibinfo {year} {2025}),\ \Eprint
	{https://arxiv.org/abs/2504.02385} {arXiv:2504.02385 [quant-ph]} \BibitemShut
	{NoStop}%
	\bibitem [{\citenamefont {Niwa}\ \emph {et~al.}(2025)\citenamefont {Niwa},
		\citenamefont {Rossi}, \citenamefont {Taranto},\ and\ \citenamefont
		{Murao}}]{niwa2025singularvaluetransformationunknown}%
	\BibitemOpen
	\bibfield  {author} {\bibinfo {author} {\bibfnamefont {R.}~\bibnamefont
			{Niwa}}, \bibinfo {author} {\bibfnamefont {Z.~M.}\ \bibnamefont {Rossi}},
		\bibinfo {author} {\bibfnamefont {P.}~\bibnamefont {Taranto}},\ and\ \bibinfo
		{author} {\bibfnamefont {M.}~\bibnamefont {Murao}},\ }\href@noop {} {\bibinfo
		{title} {Singular value transformation for unknown quantum channels}}
	(\bibinfo {year} {2025}),\ \Eprint {https://arxiv.org/abs/2506.24112}
	{arXiv:2506.24112 [quant-ph]} \BibitemShut {NoStop}%
	\bibitem [{\citenamefont {Low}\ and\ \citenamefont
		{Chuang}(2017)}]{Low2017optimal}%
	\BibitemOpen
	\bibfield  {author} {\bibinfo {author} {\bibfnamefont {G.~H.}\ \bibnamefont
			{Low}}\ and\ \bibinfo {author} {\bibfnamefont {I.~L.}\ \bibnamefont
			{Chuang}},\ }\bibfield  {title} {\bibinfo {title} {Optimal {Hamiltonian}
			simulation by quantum signal processing},\ }\href
	{https://doi.org/10.1103/PhysRevLett.118.010501} {\bibfield  {journal}
		{\bibinfo  {journal} {Phys. Rev. Lett.}\ }\textbf {\bibinfo {volume} {118}},\
		\bibinfo {pages} {010501} (\bibinfo {year} {2017})}\BibitemShut {NoStop}%
	\bibitem [{\citenamefont {Low}\ and\ \citenamefont
		{Chuang}(2019)}]{Low2019hamiltonian}%
	\BibitemOpen
	\bibfield  {author} {\bibinfo {author} {\bibfnamefont {G.~H.}\ \bibnamefont
			{Low}}\ and\ \bibinfo {author} {\bibfnamefont {I.~L.}\ \bibnamefont
			{Chuang}},\ }\bibfield  {title} {\bibinfo {title} {Hamiltonian {S}imulation
			by {Q}ubitization},\ }\href {https://doi.org/10.22331/q-2019-07-12-163}
	{\bibfield  {journal} {\bibinfo  {journal} {{Quantum}}\ }\textbf {\bibinfo
			{volume} {3}},\ \bibinfo {pages} {163} (\bibinfo {year} {2019})}\BibitemShut
	{NoStop}%
	\bibitem [{\citenamefont {An}\ \emph {et~al.}(2023{\natexlab{a}})\citenamefont
		{An}, \citenamefont {Liu},\ and\ \citenamefont {Lin}}]{ALL23}%
	\BibitemOpen
	\bibfield  {author} {\bibinfo {author} {\bibfnamefont {D.}~\bibnamefont
			{An}}, \bibinfo {author} {\bibfnamefont {J.-P.}\ \bibnamefont {Liu}},\ and\
		\bibinfo {author} {\bibfnamefont {L.}~\bibnamefont {Lin}},\ }\bibfield
	{title} {\bibinfo {title} {Linear combination of {H}amiltonian simulation for
			nonunitary dynamics with optimal state preparation cost},\ }\href@noop {}
	{\bibfield  {journal} {\bibinfo  {journal} {Physical Review Letters}\
		}\textbf {\bibinfo {volume} {131}},\ \bibinfo {pages} {150603} (\bibinfo
		{year} {2023}{\natexlab{a}})},\ \bibinfo {note}
	{\href{https://arxiv.org/abs/2303.01029}{arXiv:2303.01029}}\BibitemShut
	{NoStop}%
	\bibitem [{\citenamefont {An}\ \emph {et~al.}(2023{\natexlab{b}})\citenamefont
		{An}, \citenamefont {Childs},\ and\ \citenamefont {Lin}}]{ACL23}%
	\BibitemOpen
	\bibfield  {author} {\bibinfo {author} {\bibfnamefont {D.}~\bibnamefont
			{An}}, \bibinfo {author} {\bibfnamefont {A.~M.}\ \bibnamefont {Childs}},\
		and\ \bibinfo {author} {\bibfnamefont {L.}~\bibnamefont {Lin}},\ }\href@noop
	{} {\bibinfo {title} {Quantum algorithm for linear non-unitary dynamics with
			near-optimal dependence on all parameters}} (\bibinfo {year}
	{2023}{\natexlab{b}}),\ \bibinfo {note}
	{\href{https://arxiv.org/abs/2312.03916}{arXiv:2312.03916}}\BibitemShut
	{NoStop}%
	\bibitem [{\citenamefont {An}\ \emph {et~al.}(2024)\citenamefont {An},
		\citenamefont {Childs}, \citenamefont {Lin},\ and\ \citenamefont
		{Ying}}]{ACLY24}%
	\BibitemOpen
	\bibfield  {author} {\bibinfo {author} {\bibfnamefont {D.}~\bibnamefont
			{An}}, \bibinfo {author} {\bibfnamefont {A.~M.}\ \bibnamefont {Childs}},
		\bibinfo {author} {\bibfnamefont {L.}~\bibnamefont {Lin}},\ and\ \bibinfo
		{author} {\bibfnamefont {L.}~\bibnamefont {Ying}},\ }\href@noop {} {\bibinfo
		{title} {Laplace transform based quantum eigenvalue transformation via linear
			combination of {H}amiltonian simulation}} (\bibinfo {year} {2024}),\ \bibinfo
	{note}
	{\href{https://arxiv.org/abs/2411.04010}{arXiv:2411.04010}}\BibitemShut
	{NoStop}%
	\bibitem [{\citenamefont {Lu}\ \emph {et~al.}(2025)\citenamefont {Lu},
		\citenamefont {Li}, \citenamefont {Liu},\ and\ \citenamefont {Liu}}]{LLLL25}%
	\BibitemOpen
	\bibfield  {author} {\bibinfo {author} {\bibfnamefont {R.}~\bibnamefont
			{Lu}}, \bibinfo {author} {\bibfnamefont {H.-E.}\ \bibnamefont {Li}}, \bibinfo
		{author} {\bibfnamefont {Z.}~\bibnamefont {Liu}},\ and\ \bibinfo {author}
		{\bibfnamefont {J.-P.}\ \bibnamefont {Liu}},\ }\href@noop {} {\bibinfo
		{title} {Infinite-dimensional extension of the linear combination of
			hamiltonian simulation: Theorems and applications}} (\bibinfo {year}
	{2025}),\ \Eprint {https://arxiv.org/abs/2502.19688} {arXiv:2502.19688
		[quant-ph]} \BibitemShut {NoStop}%
	\bibitem [{\citenamefont {Jin}\ \emph {et~al.}(2023)\citenamefont {Jin},
		\citenamefont {Liu},\ and\ \citenamefont {Yu}}]{JLY23}%
	\BibitemOpen
	\bibfield  {author} {\bibinfo {author} {\bibfnamefont {S.}~\bibnamefont
			{Jin}}, \bibinfo {author} {\bibfnamefont {N.}~\bibnamefont {Liu}},\ and\
		\bibinfo {author} {\bibfnamefont {Y.}~\bibnamefont {Yu}},\ }\bibfield
	{title} {\bibinfo {title} {Quantum simulation of partial differential
			equations: Applications and detailed analysis},\ }\href@noop {} {\bibfield
		{journal} {\bibinfo  {journal} {Physical Review A}\ }\textbf {\bibinfo
			{volume} {108}},\ \bibinfo {pages} {032603} (\bibinfo {year}
		{2023})}\BibitemShut {NoStop}%
	\bibitem [{\citenamefont {Jin}\ \emph {et~al.}(2024{\natexlab{a}})\citenamefont
		{Jin}, \citenamefont {Li}, \citenamefont {Liu},\ and\ \citenamefont
		{Yu}}]{JLLY23}%
	\BibitemOpen
	\bibfield  {author} {\bibinfo {author} {\bibfnamefont {S.}~\bibnamefont
			{Jin}}, \bibinfo {author} {\bibfnamefont {X.}~\bibnamefont {Li}}, \bibinfo
		{author} {\bibfnamefont {N.}~\bibnamefont {Liu}},\ and\ \bibinfo {author}
		{\bibfnamefont {Y.}~\bibnamefont {Yu}},\ }\bibfield  {title} {\bibinfo
		{title} {Quantum simulation for quantum dynamics with artificial boundary
			conditions},\ }\href@noop {} {\bibfield  {journal} {\bibinfo  {journal} {SIAM
				Journal on Scientific Computing}\ }\textbf {\bibinfo {volume} {46}},\
		\bibinfo {pages} {B403} (\bibinfo {year} {2024}{\natexlab{a}})},\ \bibinfo
	{note}
	{\href{https://arxiv.org/abs/2304.00667}{arXiv:2304.00667}}\BibitemShut
	{NoStop}%
	\bibitem [{\citenamefont {Jin}\ \emph {et~al.}(2024{\natexlab{b}})\citenamefont
		{Jin}, \citenamefont {Li}, \citenamefont {Liu},\ and\ \citenamefont
		{Yu}}]{JLLY24}%
	\BibitemOpen
	\bibfield  {author} {\bibinfo {author} {\bibfnamefont {S.}~\bibnamefont
			{Jin}}, \bibinfo {author} {\bibfnamefont {X.}~\bibnamefont {Li}}, \bibinfo
		{author} {\bibfnamefont {N.}~\bibnamefont {Liu}},\ and\ \bibinfo {author}
		{\bibfnamefont {Y.}~\bibnamefont {Yu}},\ }\bibfield  {title} {\bibinfo
		{title} {Quantum simulation for partial differential equations with physical
			boundary or interface conditions},\ }\href@noop {} {\bibfield  {journal}
		{\bibinfo  {journal} {Journal of Computational Physics}\ }\textbf {\bibinfo
			{volume} {498}},\ \bibinfo {pages} {112707} (\bibinfo {year}
		{2024}{\natexlab{b}})},\ \bibinfo {note}
	{\href{https://arxiv.org/abs/2305.02710}{arXiv:2305.02710}}\BibitemShut
	{NoStop}%
	\bibitem [{\citenamefont {Liu}\ and\ \citenamefont
		{Wei}(2025)}]{liu2024separation}%
	\BibitemOpen
	\bibfield  {author} {\bibinfo {author} {\bibfnamefont {Z.}~\bibnamefont
			{Liu}}\ and\ \bibinfo {author} {\bibfnamefont {F.}~\bibnamefont {Wei}},\
	}\bibfield  {title} {\bibinfo {title} {Separation between entanglement
			criteria and entanglement detection protocols},\ }\href
	{https://doi.org/10.1103/hm8j-wgqm} {\bibfield  {journal} {\bibinfo
			{journal} {Phys. Rev. Res.}\ }\textbf {\bibinfo {volume} {7}},\ \bibinfo
		{pages} {033121} (\bibinfo {year} {2025})}\BibitemShut {NoStop}%
	\bibitem [{\citenamefont {Martyn}\ \emph {et~al.}(2021)\citenamefont {Martyn},
		\citenamefont {Rossi}, \citenamefont {Tan},\ and\ \citenamefont
		{Chuang}}]{martyn2021grand}%
	\BibitemOpen
	\bibfield  {author} {\bibinfo {author} {\bibfnamefont {J.~M.}\ \bibnamefont
			{Martyn}}, \bibinfo {author} {\bibfnamefont {Z.~M.}\ \bibnamefont {Rossi}},
		\bibinfo {author} {\bibfnamefont {A.~K.}\ \bibnamefont {Tan}},\ and\ \bibinfo
		{author} {\bibfnamefont {I.~L.}\ \bibnamefont {Chuang}},\ }\bibfield  {title}
	{\bibinfo {title} {Grand unification of quantum algorithms},\ }\href
	{https://doi.org/10.1103/PRXQuantum.2.040203} {\bibfield  {journal} {\bibinfo
			{journal} {PRX Quantum}\ }\textbf {\bibinfo {volume} {2}},\ \bibinfo {pages}
		{040203} (\bibinfo {year} {2021})}\BibitemShut {NoStop}%
	\bibitem [{\citenamefont {Lin}(2022)}]{lin2022lecture}%
	\BibitemOpen
	\bibfield  {author} {\bibinfo {author} {\bibfnamefont {L.}~\bibnamefont
			{Lin}},\ }\bibfield  {title} {\bibinfo {title} {Lecture notes on quantum
			algorithms for scientific computation},\ }\href
	{https://arxiv.org/abs/2201.08309} {\bibfield  {journal} {\bibinfo  {journal}
			{arXiv:2201.08309}\ } (\bibinfo {year} {2022})}\BibitemShut {NoStop}%
	\bibitem [{\citenamefont {Harrow}\ and\ \citenamefont
		{Montanaro}(2017)}]{Harrow2017supremacy}%
	\BibitemOpen
	\bibfield  {author} {\bibinfo {author} {\bibfnamefont {A.~W.}\ \bibnamefont
			{Harrow}}\ and\ \bibinfo {author} {\bibfnamefont {A.}~\bibnamefont
			{Montanaro}},\ }\bibfield  {title} {\bibinfo {title} {Quantum computational
			supremacy},\ }\href {https://doi.org/10.1038/nature23458} {\bibfield
		{journal} {\bibinfo  {journal} {Nature}\ }\textbf {\bibinfo {volume} {549}},\
		\bibinfo {pages} {203} (\bibinfo {year} {2017})}\BibitemShut {NoStop}%
	\bibitem [{\citenamefont {Chitambar}\ and\ \citenamefont
		{Gour}(2019)}]{Chitambar2019resource}%
	\BibitemOpen
	\bibfield  {author} {\bibinfo {author} {\bibfnamefont {E.}~\bibnamefont
			{Chitambar}}\ and\ \bibinfo {author} {\bibfnamefont {G.}~\bibnamefont
			{Gour}},\ }\bibfield  {title} {\bibinfo {title} {Quantum resource theories},\
	}\href {https://doi.org/10.1103/RevModPhys.91.025001} {\bibfield  {journal}
		{\bibinfo  {journal} {Rev. Mod. Phys.}\ }\textbf {\bibinfo {volume} {91}},\
		\bibinfo {pages} {025001} (\bibinfo {year} {2019})}\BibitemShut {NoStop}%
	\bibitem [{\citenamefont {Ekert}(1991)}]{ekert1991quantum}%
	\BibitemOpen
	\bibfield  {author} {\bibinfo {author} {\bibfnamefont {A.~K.}\ \bibnamefont
			{Ekert}},\ }\bibfield  {title} {\bibinfo {title} {Quantum cryptography based
			on {Bell's} theorem},\ }\href {https://doi.org/10.1103/PhysRevLett.67.661}
	{\bibfield  {journal} {\bibinfo  {journal} {Phys. Rev. Lett.}\ }\textbf
		{\bibinfo {volume} {67}},\ \bibinfo {pages} {661} (\bibinfo {year}
		{1991})}\BibitemShut {NoStop}%
	\bibitem [{\citenamefont {Bennett}\ and\ \citenamefont
		{Wiesner}(1992)}]{Bennett1992Communication}%
	\BibitemOpen
	\bibfield  {author} {\bibinfo {author} {\bibfnamefont {C.~H.}\ \bibnamefont
			{Bennett}}\ and\ \bibinfo {author} {\bibfnamefont {S.~J.}\ \bibnamefont
			{Wiesner}},\ }\bibfield  {title} {\bibinfo {title} {Communication via one-
			and two-particle operators on {Einstein-Podolsky-Rosen} states},\ }\href
	{https://doi.org/10.1103/PhysRevLett.69.2881} {\bibfield  {journal} {\bibinfo
			{journal} {Phys. Rev. Lett.}\ }\textbf {\bibinfo {volume} {69}},\ \bibinfo
		{pages} {2881} (\bibinfo {year} {1992})}\BibitemShut {NoStop}%
	\bibitem [{\citenamefont {Bennett}\ \emph {et~al.}(1993)\citenamefont
		{Bennett}, \citenamefont {Brassard}, \citenamefont {Cr\'epeau}, \citenamefont
		{Jozsa}, \citenamefont {Peres},\ and\ \citenamefont
		{Wootters}}]{Bennett1993Teleporting}%
	\BibitemOpen
	\bibfield  {author} {\bibinfo {author} {\bibfnamefont {C.~H.}\ \bibnamefont
			{Bennett}}, \bibinfo {author} {\bibfnamefont {G.}~\bibnamefont {Brassard}},
		\bibinfo {author} {\bibfnamefont {C.}~\bibnamefont {Cr\'epeau}}, \bibinfo
		{author} {\bibfnamefont {R.}~\bibnamefont {Jozsa}}, \bibinfo {author}
		{\bibfnamefont {A.}~\bibnamefont {Peres}},\ and\ \bibinfo {author}
		{\bibfnamefont {W.~K.}\ \bibnamefont {Wootters}},\ }\bibfield  {title}
	{\bibinfo {title} {Teleporting an unknown quantum state via dual classical
			and {Einstein-Podolsky-Rosen} channels},\ }\href
	{https://doi.org/10.1103/PhysRevLett.70.1895} {\bibfield  {journal} {\bibinfo
			{journal} {Phys. Rev. Lett.}\ }\textbf {\bibinfo {volume} {70}},\ \bibinfo
		{pages} {1895} (\bibinfo {year} {1993})}\BibitemShut {NoStop}%
	\bibitem [{\citenamefont {Raussendorf}\ and\ \citenamefont
		{Briegel}(2001)}]{Raussendorf2001OneWay}%
	\BibitemOpen
	\bibfield  {author} {\bibinfo {author} {\bibfnamefont {R.}~\bibnamefont
			{Raussendorf}}\ and\ \bibinfo {author} {\bibfnamefont {H.~J.}\ \bibnamefont
			{Briegel}},\ }\bibfield  {title} {\bibinfo {title} {A one-way quantum
			computer},\ }\href {https://doi.org/10.1103/PhysRevLett.86.5188} {\bibfield
		{journal} {\bibinfo  {journal} {Phys. Rev. Lett.}\ }\textbf {\bibinfo
			{volume} {86}},\ \bibinfo {pages} {5188} (\bibinfo {year}
		{2001})}\BibitemShut {NoStop}%
	\bibitem [{\citenamefont {Calabrese}\ and\ \citenamefont
		{Cardy}(2004)}]{Pasquale2004qft}%
	\BibitemOpen
	\bibfield  {author} {\bibinfo {author} {\bibfnamefont {P.}~\bibnamefont
			{Calabrese}}\ and\ \bibinfo {author} {\bibfnamefont {J.}~\bibnamefont
			{Cardy}},\ }\bibfield  {title} {\bibinfo {title} {Entanglement entropy and
			quantum field theory},\ }\href
	{https://doi.org/10.1088/1742-5468/2004/06/P06002} {\bibfield  {journal}
		{\bibinfo  {journal} {Journal of Statistical Mechanics: Theory and
				Experiment}\ }\textbf {\bibinfo {volume} {2004}},\ \bibinfo {pages} {P06002}
		(\bibinfo {year} {2004})}\BibitemShut {NoStop}%
	\bibitem [{\citenamefont {Amico}\ \emph {et~al.}(2008)\citenamefont {Amico},
		\citenamefont {Fazio}, \citenamefont {Osterloh},\ and\ \citenamefont
		{Vedral}}]{Amico2008Entanglement}%
	\BibitemOpen
	\bibfield  {author} {\bibinfo {author} {\bibfnamefont {L.}~\bibnamefont
			{Amico}}, \bibinfo {author} {\bibfnamefont {R.}~\bibnamefont {Fazio}},
		\bibinfo {author} {\bibfnamefont {A.}~\bibnamefont {Osterloh}},\ and\
		\bibinfo {author} {\bibfnamefont {V.}~\bibnamefont {Vedral}},\ }\bibfield
	{title} {\bibinfo {title} {Entanglement in many-body systems},\ }\href
	{https://doi.org/10.1103/RevModPhys.80.517} {\bibfield  {journal} {\bibinfo
			{journal} {Rev. Mod. Phys.}\ }\textbf {\bibinfo {volume} {80}},\ \bibinfo
		{pages} {517} (\bibinfo {year} {2008})}\BibitemShut {NoStop}%
	\bibitem [{\citenamefont {Horodecki}\ \emph {et~al.}(2009)\citenamefont
		{Horodecki}, \citenamefont {Horodecki}, \citenamefont {Horodecki},\ and\
		\citenamefont {Horodecki}}]{Horodecki2009entanglement}%
	\BibitemOpen
	\bibfield  {author} {\bibinfo {author} {\bibfnamefont {R.}~\bibnamefont
			{Horodecki}}, \bibinfo {author} {\bibfnamefont {P.}~\bibnamefont
			{Horodecki}}, \bibinfo {author} {\bibfnamefont {M.}~\bibnamefont
			{Horodecki}},\ and\ \bibinfo {author} {\bibfnamefont {K.}~\bibnamefont
			{Horodecki}},\ }\bibfield  {title} {\bibinfo {title} {Quantum entanglement},\
	}\href {https://doi.org/10.1103/RevModPhys.81.865} {\bibfield  {journal}
		{\bibinfo  {journal} {Rev. Mod. Phys.}\ }\textbf {\bibinfo {volume} {81}},\
		\bibinfo {pages} {865} (\bibinfo {year} {2009})}\BibitemShut {NoStop}%
	\bibitem [{\citenamefont {Pan}\ \emph {et~al.}(2012)\citenamefont {Pan},
		\citenamefont {Chen}, \citenamefont {Lu}, \citenamefont {Weinfurter},
		\citenamefont {Zeilinger},\ and\ \citenamefont {\ifmmode~\dot{Z}\else
			\.{Z}\fi{}ukowski}}]{pan2012multiphoton}%
	\BibitemOpen
	\bibfield  {author} {\bibinfo {author} {\bibfnamefont {J.-W.}\ \bibnamefont
			{Pan}}, \bibinfo {author} {\bibfnamefont {Z.-B.}\ \bibnamefont {Chen}},
		\bibinfo {author} {\bibfnamefont {C.-Y.}\ \bibnamefont {Lu}}, \bibinfo
		{author} {\bibfnamefont {H.}~\bibnamefont {Weinfurter}}, \bibinfo {author}
		{\bibfnamefont {A.}~\bibnamefont {Zeilinger}},\ and\ \bibinfo {author}
		{\bibfnamefont {M.}~\bibnamefont {\ifmmode~\dot{Z}\else \.{Z}\fi{}ukowski}},\
	}\bibfield  {title} {\bibinfo {title} {Multiphoton entanglement and
			interferometry},\ }\href {https://doi.org/10.1103/RevModPhys.84.777}
	{\bibfield  {journal} {\bibinfo  {journal} {Rev. Mod. Phys.}\ }\textbf
		{\bibinfo {volume} {84}},\ \bibinfo {pages} {777} (\bibinfo {year}
		{2012})}\BibitemShut {NoStop}%
	\bibitem [{\citenamefont {Peres}(1996)}]{peres1996ppt}%
	\BibitemOpen
	\bibfield  {author} {\bibinfo {author} {\bibfnamefont {A.}~\bibnamefont
			{Peres}},\ }\bibfield  {title} {\bibinfo {title} {Separability criterion for
			density matrices},\ }\href {https://doi.org/10.1103/PhysRevLett.77.1413}
	{\bibfield  {journal} {\bibinfo  {journal} {Phys. Rev. Lett.}\ }\textbf
		{\bibinfo {volume} {77}},\ \bibinfo {pages} {1413} (\bibinfo {year}
		{1996})}\BibitemShut {NoStop}%
	\bibitem [{\citenamefont {Horodecki}\ and\ \citenamefont
		{Horodecki}(1999)}]{Horodecki1999reduction}%
	\BibitemOpen
	\bibfield  {author} {\bibinfo {author} {\bibfnamefont {M.}~\bibnamefont
			{Horodecki}}\ and\ \bibinfo {author} {\bibfnamefont {P.}~\bibnamefont
			{Horodecki}},\ }\bibfield  {title} {\bibinfo {title} {Reduction criterion of
			separability and limits for a class of distillation protocols},\ }\href
	{https://doi.org/10.1103/PhysRevA.59.4206} {\bibfield  {journal} {\bibinfo
			{journal} {Phys. Rev. A}\ }\textbf {\bibinfo {volume} {59}},\ \bibinfo
		{pages} {4206} (\bibinfo {year} {1999})}\BibitemShut {NoStop}%
	\bibitem [{\citenamefont {Gray}\ \emph {et~al.}(2018)\citenamefont {Gray},
		\citenamefont {Banchi}, \citenamefont {Bayat},\ and\ \citenamefont
		{Bose}}]{gray2018machine}%
	\BibitemOpen
	\bibfield  {author} {\bibinfo {author} {\bibfnamefont {J.}~\bibnamefont
			{Gray}}, \bibinfo {author} {\bibfnamefont {L.}~\bibnamefont {Banchi}},
		\bibinfo {author} {\bibfnamefont {A.}~\bibnamefont {Bayat}},\ and\ \bibinfo
		{author} {\bibfnamefont {S.}~\bibnamefont {Bose}},\ }\bibfield  {title}
	{\bibinfo {title} {Machine-learning-assisted many-body entanglement
			measurement},\ }\href {https://doi.org/10.1103/PhysRevLett.121.150503}
	{\bibfield  {journal} {\bibinfo  {journal} {Phys. Rev. Lett.}\ }\textbf
		{\bibinfo {volume} {121}},\ \bibinfo {pages} {150503} (\bibinfo {year}
		{2018})}\BibitemShut {NoStop}%
	\bibitem [{\citenamefont {Yu}\ \emph {et~al.}(2021)\citenamefont {Yu},
		\citenamefont {Imai},\ and\ \citenamefont {G\"uhne}}]{yu2021optimal}%
	\BibitemOpen
	\bibfield  {author} {\bibinfo {author} {\bibfnamefont {X.-D.}\ \bibnamefont
			{Yu}}, \bibinfo {author} {\bibfnamefont {S.}~\bibnamefont {Imai}},\ and\
		\bibinfo {author} {\bibfnamefont {O.}~\bibnamefont {G\"uhne}},\ }\bibfield
	{title} {\bibinfo {title} {Optimal entanglement certification from moments of
			the partial transpose},\ }\href
	{https://doi.org/10.1103/PhysRevLett.127.060504} {\bibfield  {journal}
		{\bibinfo  {journal} {Phys. Rev. Lett.}\ }\textbf {\bibinfo {volume} {127}},\
		\bibinfo {pages} {060504} (\bibinfo {year} {2021})}\BibitemShut {NoStop}%
	\bibitem [{\citenamefont {Zhou}\ \emph {et~al.}(2020)\citenamefont {Zhou},
		\citenamefont {Zeng},\ and\ \citenamefont {Liu}}]{zhou2020Single}%
	\BibitemOpen
	\bibfield  {author} {\bibinfo {author} {\bibfnamefont {Y.}~\bibnamefont
			{Zhou}}, \bibinfo {author} {\bibfnamefont {P.}~\bibnamefont {Zeng}},\ and\
		\bibinfo {author} {\bibfnamefont {Z.}~\bibnamefont {Liu}},\ }\bibfield
	{title} {\bibinfo {title} {Single-copies estimation of entanglement
			negativity},\ }\href {https://doi.org/10.1103/PhysRevLett.125.200502}
	{\bibfield  {journal} {\bibinfo  {journal} {Phys. Rev. Lett.}\ }\textbf
		{\bibinfo {volume} {125}},\ \bibinfo {pages} {200502} (\bibinfo {year}
		{2020})}\BibitemShut {NoStop}%
	\bibitem [{\citenamefont {Elben}\ \emph {et~al.}(2020)\citenamefont {Elben},
		\citenamefont {Kueng}, \citenamefont {Huang}, \citenamefont {van Bijnen},
		\citenamefont {Kokail}, \citenamefont {Dalmonte}, \citenamefont {Calabrese},
		\citenamefont {Kraus}, \citenamefont {Preskill}, \citenamefont {Zoller},\
		and\ \citenamefont {Vermersch}}]{elben2020mixed}%
	\BibitemOpen
	\bibfield  {author} {\bibinfo {author} {\bibfnamefont {A.}~\bibnamefont
			{Elben}}, \bibinfo {author} {\bibfnamefont {R.}~\bibnamefont {Kueng}},
		\bibinfo {author} {\bibfnamefont {H.-Y.~R.}\ \bibnamefont {Huang}}, \bibinfo
		{author} {\bibfnamefont {R.}~\bibnamefont {van Bijnen}}, \bibinfo {author}
		{\bibfnamefont {C.}~\bibnamefont {Kokail}}, \bibinfo {author} {\bibfnamefont
			{M.}~\bibnamefont {Dalmonte}}, \bibinfo {author} {\bibfnamefont
			{P.}~\bibnamefont {Calabrese}}, \bibinfo {author} {\bibfnamefont
			{B.}~\bibnamefont {Kraus}}, \bibinfo {author} {\bibfnamefont
			{J.}~\bibnamefont {Preskill}}, \bibinfo {author} {\bibfnamefont
			{P.}~\bibnamefont {Zoller}},\ and\ \bibinfo {author} {\bibfnamefont
			{B.}~\bibnamefont {Vermersch}},\ }\bibfield  {title} {\bibinfo {title}
		{Mixed-state entanglement from local randomized measurements},\ }\href
	{https://doi.org/10.1103/PhysRevLett.125.200501} {\bibfield  {journal}
		{\bibinfo  {journal} {Phys. Rev. Lett.}\ }\textbf {\bibinfo {volume} {125}},\
		\bibinfo {pages} {200501} (\bibinfo {year} {2020})}\BibitemShut {NoStop}%
	\bibitem [{\citenamefont {Quintino}\ \emph {et~al.}(2019)\citenamefont
		{Quintino}, \citenamefont {Dong}, \citenamefont {Shimbo}, \citenamefont
		{Soeda},\ and\ \citenamefont {Murao}}]{Quintino2019Reversing}%
	\BibitemOpen
	\bibfield  {author} {\bibinfo {author} {\bibfnamefont {M.~T.}\ \bibnamefont
			{Quintino}}, \bibinfo {author} {\bibfnamefont {Q.}~\bibnamefont {Dong}},
		\bibinfo {author} {\bibfnamefont {A.}~\bibnamefont {Shimbo}}, \bibinfo
		{author} {\bibfnamefont {A.}~\bibnamefont {Soeda}},\ and\ \bibinfo {author}
		{\bibfnamefont {M.}~\bibnamefont {Murao}},\ }\bibfield  {title} {\bibinfo
		{title} {Reversing unknown quantum transformations: Universal quantum circuit
			for inverting general unitary operations},\ }\href
	{https://doi.org/10.1103/PhysRevLett.123.210502} {\bibfield  {journal}
		{\bibinfo  {journal} {Phys. Rev. Lett.}\ }\textbf {\bibinfo {volume} {123}},\
		\bibinfo {pages} {210502} (\bibinfo {year} {2019})}\BibitemShut {NoStop}%
	\bibitem [{\citenamefont {Yoshida}\ \emph {et~al.}(2023)\citenamefont
		{Yoshida}, \citenamefont {Soeda},\ and\ \citenamefont
		{Murao}}]{Yoshida2023Reversing}%
	\BibitemOpen
	\bibfield  {author} {\bibinfo {author} {\bibfnamefont {S.}~\bibnamefont
			{Yoshida}}, \bibinfo {author} {\bibfnamefont {A.}~\bibnamefont {Soeda}},\
		and\ \bibinfo {author} {\bibfnamefont {M.}~\bibnamefont {Murao}},\ }\bibfield
	{title} {\bibinfo {title} {Reversing unknown qubit-unitary operation,
			deterministically and exactly},\ }\href
	{https://doi.org/10.1103/PhysRevLett.131.120602} {\bibfield  {journal}
		{\bibinfo  {journal} {Phys. Rev. Lett.}\ }\textbf {\bibinfo {volume} {131}},\
		\bibinfo {pages} {120602} (\bibinfo {year} {2023})}\BibitemShut {NoStop}%
	\bibitem [{\citenamefont {Zhu}\ \emph {et~al.}(2024)\citenamefont {Zhu},
		\citenamefont {Mo}, \citenamefont {Chen},\ and\ \citenamefont
		{Wang}}]{Zhu2024Reversing}%
	\BibitemOpen
	\bibfield  {author} {\bibinfo {author} {\bibfnamefont {C.}~\bibnamefont
			{Zhu}}, \bibinfo {author} {\bibfnamefont {Y.}~\bibnamefont {Mo}}, \bibinfo
		{author} {\bibfnamefont {Y.-A.}\ \bibnamefont {Chen}},\ and\ \bibinfo
		{author} {\bibfnamefont {X.}~\bibnamefont {Wang}},\ }\bibfield  {title}
	{\bibinfo {title} {Reversing unknown quantum processes via virtual combs for
			channels with limited information},\ }\href
	{https://doi.org/10.1103/PhysRevLett.133.030801} {\bibfield  {journal}
		{\bibinfo  {journal} {Phys. Rev. Lett.}\ }\textbf {\bibinfo {volume} {133}},\
		\bibinfo {pages} {030801} (\bibinfo {year} {2024})}\BibitemShut {NoStop}%
	\bibitem [{\citenamefont {Mo}\ \emph {et~al.}(2025{\natexlab{a}})\citenamefont
		{Mo}, \citenamefont {Lin},\ and\ \citenamefont
		{Wang}}]{mo2025efficientinversionunknownunitary}%
	\BibitemOpen
	\bibfield  {author} {\bibinfo {author} {\bibfnamefont {Y.}~\bibnamefont
			{Mo}}, \bibinfo {author} {\bibfnamefont {T.}~\bibnamefont {Lin}},\ and\
		\bibinfo {author} {\bibfnamefont {X.}~\bibnamefont {Wang}},\ }\href@noop {}
	{\bibinfo {title} {Efficient inversion of unknown unitary operations with
			structured {Hamiltonians}}} (\bibinfo {year} {2025}{\natexlab{a}}),\ \Eprint
	{https://arxiv.org/abs/2506.20570} {arXiv:2506.20570 [quant-ph]} \BibitemShut
	{NoStop}%
	\bibitem [{\citenamefont {Mo}\ \emph {et~al.}(2025{\natexlab{b}})\citenamefont
		{Mo}, \citenamefont {Zhang}, \citenamefont {Chen}, \citenamefont {Liu},
		\citenamefont {Lin},\ and\ \citenamefont {Wang}}]{Mo2025Parameterized}%
	\BibitemOpen
	\bibfield  {author} {\bibinfo {author} {\bibfnamefont {Y.}~\bibnamefont
			{Mo}}, \bibinfo {author} {\bibfnamefont {L.}~\bibnamefont {Zhang}}, \bibinfo
		{author} {\bibfnamefont {Y.-A.}\ \bibnamefont {Chen}}, \bibinfo {author}
		{\bibfnamefont {Y.}~\bibnamefont {Liu}}, \bibinfo {author} {\bibfnamefont
			{T.}~\bibnamefont {Lin}},\ and\ \bibinfo {author} {\bibfnamefont
			{X.}~\bibnamefont {Wang}},\ }\bibfield  {title} {\bibinfo {title}
		{Parameterized quantum comb and simpler circuits for reversing unknown
			qubit-unitary operations},\ }\href
	{https://doi.org/10.1038/s41534-025-00979-1} {\bibfield  {journal} {\bibinfo
			{journal} {npj Quantum Information}\ }\textbf {\bibinfo {volume} {11}},\
		\bibinfo {pages} {32} (\bibinfo {year} {2025}{\natexlab{b}})}\BibitemShut
	{NoStop}%
	\bibitem [{\citenamefont {Shor}(1995)}]{shor1995scheme}%
	\BibitemOpen
	\bibfield  {author} {\bibinfo {author} {\bibfnamefont {P.~W.}\ \bibnamefont
			{Shor}},\ }\bibfield  {title} {\bibinfo {title} {Scheme for reducing
			decoherence in quantum computer memory},\ }\href
	{https://doi.org/10.1103/PhysRevA.52.R2493} {\bibfield  {journal} {\bibinfo
			{journal} {Phys. Rev. A}\ }\textbf {\bibinfo {volume} {52}},\ \bibinfo
		{pages} {R2493} (\bibinfo {year} {1995})}\BibitemShut {NoStop}%
	\bibitem [{\citenamefont {Gottesman}(1997)}]{gottesman1997stabilizer}%
	\BibitemOpen
	\bibfield  {author} {\bibinfo {author} {\bibfnamefont {D.}~\bibnamefont
			{Gottesman}},\ }\href@noop {} {\emph {\bibinfo {title} {Stabilizer codes and
				quantum error correction}}}\ (\bibinfo  {publisher} {California Institute of
		Technology},\ \bibinfo {year} {1997})\BibitemShut {NoStop}%
	\bibitem [{\citenamefont {Jiang}\ \emph {et~al.}(2021)\citenamefont {Jiang},
		\citenamefont {Wang},\ and\ \citenamefont {Wang}}]{Jiang2021physical}%
	\BibitemOpen
	\bibfield  {author} {\bibinfo {author} {\bibfnamefont {J.}~\bibnamefont
			{Jiang}}, \bibinfo {author} {\bibfnamefont {K.}~\bibnamefont {Wang}},\ and\
		\bibinfo {author} {\bibfnamefont {X.}~\bibnamefont {Wang}},\ }\bibfield
	{title} {\bibinfo {title} {Physical {I}mplementability of {L}inear {M}aps and
			{I}ts {A}pplication in {E}rror {M}itigation},\ }\href
	{https://doi.org/10.22331/q-2021-12-07-600} {\bibfield  {journal} {\bibinfo
			{journal} {{Quantum}}\ }\textbf {\bibinfo {volume} {5}},\ \bibinfo {pages}
		{600} (\bibinfo {year} {2021})}\BibitemShut {NoStop}%
	\bibitem [{\citenamefont {Regula}\ \emph {et~al.}(2021)\citenamefont {Regula},
		\citenamefont {Takagi},\ and\ \citenamefont {Gu}}]{Regula2021operational}%
	\BibitemOpen
	\bibfield  {author} {\bibinfo {author} {\bibfnamefont {B.}~\bibnamefont
			{Regula}}, \bibinfo {author} {\bibfnamefont {R.}~\bibnamefont {Takagi}},\
		and\ \bibinfo {author} {\bibfnamefont {M.}~\bibnamefont {Gu}},\ }\bibfield
	{title} {\bibinfo {title} {Operational applications of the diamond norm and
			related measures in quantifying the non-physicality of quantum maps},\ }\href
	{https://doi.org/10.22331/q-2021-08-09-522} {\bibfield  {journal} {\bibinfo
			{journal} {{Quantum}}\ }\textbf {\bibinfo {volume} {5}},\ \bibinfo {pages}
		{522} (\bibinfo {year} {2021})}\BibitemShut {NoStop}%
	\bibitem [{\citenamefont {Berry}\ \emph {et~al.}(2014)\citenamefont {Berry},
		\citenamefont {Childs}, \citenamefont {Cleve}, \citenamefont {Kothari},\ and\
		\citenamefont {Somma}}]{Berry2014Exponential}%
	\BibitemOpen
	\bibfield  {author} {\bibinfo {author} {\bibfnamefont {D.~W.}\ \bibnamefont
			{Berry}}, \bibinfo {author} {\bibfnamefont {A.~M.}\ \bibnamefont {Childs}},
		\bibinfo {author} {\bibfnamefont {R.}~\bibnamefont {Cleve}}, \bibinfo
		{author} {\bibfnamefont {R.}~\bibnamefont {Kothari}},\ and\ \bibinfo {author}
		{\bibfnamefont {R.~D.}\ \bibnamefont {Somma}},\ }\bibfield  {title} {\bibinfo
		{title} {Exponential improvement in precision for simulating sparse
			hamiltonians},\ }in\ \href {https://doi.org/10.1145/2591796.2591854} {\emph
		{\bibinfo {booktitle} {Proceedings of the Forty-Sixth Annual ACM Symposium on
				Theory of Computing}}},\ \bibinfo {series and number} {STOC '14}\ (\bibinfo
	{publisher} {Association for Computing Machinery},\ \bibinfo {address} {New
		York, NY, USA},\ \bibinfo {year} {2014})\ p.\ \bibinfo {pages}
	{283–292}\BibitemShut {NoStop}%
	\bibitem [{\citenamefont {Stinespring}(1955)}]{Stinespring1955Positive}%
	\BibitemOpen
	\bibfield  {author} {\bibinfo {author} {\bibfnamefont {W.~F.}\ \bibnamefont
			{Stinespring}},\ }\bibfield  {title} {\bibinfo {title} {Positive functions on
			c*-algebras},\ }\href {http://www.jstor.org/stable/2032342} {\bibfield
		{journal} {\bibinfo  {journal} {Proceedings of the American Mathematical
				Society}\ }\textbf {\bibinfo {volume} {6}},\ \bibinfo {pages} {211} (\bibinfo
		{year} {1955})}\BibitemShut {NoStop}%
	\bibitem [{\citenamefont {Wood}\ \emph {et~al.}(2011)\citenamefont {Wood},
		\citenamefont {Biamonte},\ and\ \citenamefont {Cory}}]{wood2011tensor}%
	\BibitemOpen
	\bibfield  {author} {\bibinfo {author} {\bibfnamefont {C.~J.}\ \bibnamefont
			{Wood}}, \bibinfo {author} {\bibfnamefont {J.~D.}\ \bibnamefont {Biamonte}},\
		and\ \bibinfo {author} {\bibfnamefont {D.~G.}\ \bibnamefont {Cory}},\
	}\bibfield  {title} {\bibinfo {title} {Tensor networks and graphical calculus
			for open quantum systems},\ }\href {https://arxiv.org/abs/1111.6950}
	{\bibfield  {journal} {\bibinfo  {journal} {arXiv:1111.6950}\ } (\bibinfo
		{year} {2011})}\BibitemShut {NoStop}%
	\bibitem [{\citenamefont {Li}\ \emph {et~al.}(2023)\citenamefont {Li},
		\citenamefont {Ni},\ and\ \citenamefont {Ying}}]{li2023efficient}%
	\BibitemOpen
	\bibfield  {author} {\bibinfo {author} {\bibfnamefont {H.}~\bibnamefont
			{Li}}, \bibinfo {author} {\bibfnamefont {H.}~\bibnamefont {Ni}},\ and\
		\bibinfo {author} {\bibfnamefont {L.}~\bibnamefont {Ying}},\ }\bibfield
	{title} {\bibinfo {title} {On efficient quantum block encoding of
			pseudo-differential operators},\ }\href@noop {} {\bibfield  {journal}
		{\bibinfo  {journal} {Quantum}\ }\textbf {\bibinfo {volume} {7}},\ \bibinfo
		{pages} {1031} (\bibinfo {year} {2023})}\BibitemShut {NoStop}%
	\bibitem [{\citenamefont {Dong}\ \emph {et~al.}(2022)\citenamefont {Dong},
		\citenamefont {Lin},\ and\ \citenamefont {Tong}}]{dong2022ground}%
	\BibitemOpen
	\bibfield  {author} {\bibinfo {author} {\bibfnamefont {Y.}~\bibnamefont
			{Dong}}, \bibinfo {author} {\bibfnamefont {L.}~\bibnamefont {Lin}},\ and\
		\bibinfo {author} {\bibfnamefont {Y.}~\bibnamefont {Tong}},\ }\bibfield
	{title} {\bibinfo {title} {Ground-state preparation and energy estimation on
			early fault-tolerant quantum computers via quantum eigenvalue transformation
			of unitary matrices},\ }\href {https://doi.org/10.1103/PRXQuantum.3.040305}
	{\bibfield  {journal} {\bibinfo  {journal} {PRX Quantum}\ }\textbf {\bibinfo
			{volume} {3}},\ \bibinfo {pages} {040305} (\bibinfo {year}
		{2022})}\BibitemShut {NoStop}%
	\bibitem [{\citenamefont {Lubkin}\ and\ \citenamefont
		{Lubkin}(1993)}]{Lubkin1993average}%
	\BibitemOpen
	\bibfield  {author} {\bibinfo {author} {\bibfnamefont {E.}~\bibnamefont
			{Lubkin}}\ and\ \bibinfo {author} {\bibfnamefont {T.}~\bibnamefont
			{Lubkin}},\ }\bibfield  {title} {\bibinfo {title} {Average quantal behavior
			and thermodynamic isolation},\ }\href {https://doi.org/10.1007/BF01215300}
	{\bibfield  {journal} {\bibinfo  {journal} {International Journal of
				Theoretical Physics}\ }\textbf {\bibinfo {volume} {32}},\ \bibinfo {pages}
		{933} (\bibinfo {year} {1993})}\BibitemShut {NoStop}%
	\bibitem [{\citenamefont {Hiroka}\ \emph {et~al.}(2025)\citenamefont {Hiroka},
		\citenamefont {Hsieh},\ and\ \citenamefont
		{Morimae}}]{hiroka2025hardnessquantumdistributionlearning}%
	\BibitemOpen
	\bibfield  {author} {\bibinfo {author} {\bibfnamefont {T.}~\bibnamefont
			{Hiroka}}, \bibinfo {author} {\bibfnamefont {M.-H.}\ \bibnamefont {Hsieh}},\
		and\ \bibinfo {author} {\bibfnamefont {T.}~\bibnamefont {Morimae}},\
	}\href@noop {} {\bibinfo {title} {Hardness of quantum distribution learning
			and quantum cryptography}} (\bibinfo {year} {2025}),\ \Eprint
	{https://arxiv.org/abs/2507.01292} {arXiv:2507.01292 [quant-ph]} \BibitemShut
	{NoStop}%
\end{thebibliography}

%

\onecolumngrid

\appendix

\section{Proof of Theorem~\ref{thm:main}}\label{app:proof_of_main}

We can write $\mc{N}(\rho)$ as (see Fig.~\ref{fig:topfigure}(d))
\begin{equation}
\begin{aligned}
\mc{N}(\rho)&=\mc{N}(\Tr_{\mbb{C}^{2^r}}\ketbra{\psi}{\psi})\\
&=\mc{N}\big((\bra{\psi}\otimes\mbb{I}_n)(\mbb{I}_r\otimes S)(\ket{\psi}\otimes\mbb{I}_n)\big)\\
&=\left(\bra{\psi}\otimes\mbb{I}_{k}\right)\left(\mbb{I}_r\otimes\Lambda_{\mc{N}}^{\mathrm{T}_1}\right)\left(\ket{\psi}\otimes\mbb{I}_k\right)\\
&=\left(\bra{0^{r+n}}U^{\rho\dagger}\otimes\mbb{I}_{k}\right)\left(\mbb{I}_r\otimes\Lambda_{\mc{N}}^{\mathrm{T}_1}\right)\left(U^{\rho}\ket{0^{r+n}}\otimes\mbb{I}_k\right),
\end{aligned}
\end{equation}
where $S$ is the swap operator on $\mc{H}\otimes\mc{H}$. Since $U_\mc{N}$ is an $(\alpha,m,\epsilon)$-block encoding of $\Lambda_\mc{N}^{\mathrm{T}_1}$, we have
\begin{equation}
\bignorm{\Lambda_\mc{N}^{\mathrm{T}_1}-\alpha(\bra{0^{m}}\otimes\mbb{I}_{n+k})U_{\mc{N}}(\ket{0^m}\otimes\mbb{I}_{n+k})}_\infty\le\epsilon.
\end{equation}
Combining the above two equations, we arrive at
\begin{equation}
\begin{aligned}
&\bignorm{\mc{N}(\rho)-\alpha\left(\bra{0^{m+r+n}}\otimes\mbb{I}_k\right)(\mbb{I}_m\otimes U^{\rho\dagger}\otimes\mbb{I}_{k})\left(U_{\mc{N}}\otimes\mbb{I}_r\right)(\mbb{I}_m\otimes U^{\rho}\otimes\mbb{I}_{k})\left(\ket{0^{m+r+n}}\otimes\mbb{I}_k\right)}_\infty\\
=&{\big|\big|}\left(\bra{0^{r+n}}U^{\rho\dagger}\otimes\mbb{I}_{k}\right)\left(\mbb{I}_r\otimes\Lambda_{\mc{N}}^{\mathrm{T}_1}\right)\left(U^{\rho}\ket{0^{r+n}}\otimes\mbb{I}_k\right)\\
&-\left(\bra{0^{r+n}}U^{\rho\dagger}\otimes\mbb{I}_{k}\right)\Big(\mbb{I}_{r}\otimes\alpha(\bra{0}\otimes\mbb{I}_{n+k})U_{\mc{N}}(\ket{0}\otimes\mbb{I}_{n+k})\Big)\left(U^{\rho}\ket{0^{r+n}}\otimes\mbb{I}_k\right){\big|\big|}_\infty\\
\le&\bignorm{\bra{0^{r+n}}U^{\rho\dagger}\otimes\mbb{I}_{k}}_\infty\bignorm{\mbb{I}_r\otimes\Lambda_{\mc{N}}^{\mathrm{T}_1}-\mbb{I}_{r}\otimes\alpha(\bra{0}\otimes\mbb{I}_{n+k})U_{\mc{N}}(\ket{0}\otimes\mbb{I}_{n+k})}_\infty\bignorm{U^{\rho}\ket{0^{r+n}}\otimes\mbb{I}_k}_\infty\\
=&\bignorm{\Lambda_\mc{N}^{\mathrm{T}_1}-\alpha(\bra{0^{m}}\otimes\mbb{I}_{n+k})U_{\mc{N}}(\ket{0^m}\otimes\mbb{I}_{n+k})}_\infty\le\epsilon.
\end{aligned}
\end{equation}
Therefore, $(\mbb{I}\otimes U^{\rho\dagger}\otimes\mbb{I}_{\mc{K}})\left(U_{\mc{N}}\otimes\mbb{I}_\mc{R}\right)(\mbb{I}\otimes U^{\rho}\otimes\mbb{I}_{\mc{K}})$ is an $(\alpha,m+r+n,\epsilon)$-block encoding of $\mc{N}(\rho)$.

\section{Proof of Theorem~\ref{thm:ED}}\label{app:proof_of_ED}
\begin{proof}

Let $\mathscr{P}_0=\{\psi_A\otimes\psi_B\}$ be the sets of product $n(=2q)$-qubit pure states, and $\mathscr{P}_1=\{\psi_{AB}\}$ be the sets of global pure $n$-qubit states.

Consider running the circuit shown in Fig.~\ref{fig:applications}(a) once.

If $\rho=\ketbra{\psi_A\otimes\psi_B}{\psi_A\otimes\psi_B}\in\mathscr{P}_0$, then $\mc{R}\otimes\mc{I}(\psi_A\otimes\psi_B)=(\mbb{I}_{q}-\psi_A)\otimes\psi_B$, and 
\begin{equation}
\begin{aligned}
\Pr[\text{get } \ket{0^{m+r+n}}]=\frac{1}{4}\Tr[\big((\mbb{I}_{q}-\psi_A)\otimes\psi_B\big)\psi_A\otimes\psi_B\big((\mbb{I}_{q}-\psi_A)\otimes\psi_B\big)]=0.
\end{aligned}
\end{equation}

If $\rho=\ketbra{\psi}{\psi}\in\mathscr{P}_1$, we have $\mc{R}\otimes\mc{I}(\psi)=\mbb{I}_q\otimes\Tr_A(\psi)-\psi$.
The average purity of the reduced density matrix $\Tr_A(\psi)$ is known to satisfy \cite{Lubkin1993average}
\begin{equation}
\underset{\psi\sim\text{Haar}}{\mathbb{E}}\Tr\left[\Tr_A(\psi)^2\right]=\frac{2\sqrt{d}}{d+1}.
\end{equation}
Note that
\begin{equation}
\begin{aligned}
\underset{\psi\sim\text{Haar}}{\mathbb{E}}\norm{\Tr_A(\psi)}_{\infty}\le&\underset{\psi\sim\text{Haar}}{\mathbb{E}}\sqrt{\Tr\left[\Tr_A(\psi)^2\right]}\\
\le&\sqrt{\underset{\psi\sim\text{Haar}}{\mathbb{E}}\Tr\left[\Tr_A(\psi)^2\right]}, 
\end{aligned}
\end{equation}
where the second inequality follows from Jensen's inequality. Thus, we have
\begin{equation}
\underset{\psi\sim\text{Haar}}{\mathbb{E}}\norm{\Tr_A(\psi)}_{\infty}\rightarrow0,\text{ as }d\rightarrow\infty.
\end{equation}
When $d$ is sufficiently large, by Markov's inequality, we can guarantee $\norm{\Tr_A(\psi)}_{\infty}\le0.01$ with a probability of at least $0.99$.
When $\norm{\Tr_A(\psi)}_{\infty}\le0.01$, we have
\begin{equation}
\begin{aligned}
&\Big|\Tr[\mc{R}\otimes\mc{I}(\psi)\psi\mc{R}\otimes\mc{I}(\psi)]-1\Big|\\
=&\Big|\Tr[\mc{R}\otimes\mc{I}(\psi)\psi\mc{R}\otimes\mc{I}(\psi)]-\Tr[(-\psi)\psi(-\psi)]\Big|\\
=&\Big|\Tr[\mc{R}\otimes\mc{I}(\psi)\psi\big(\mc{R}\otimes\mc{I}(\psi)+\psi\big)]+\Tr[\big(\mc{R}\otimes\mc{I}(\psi)+\psi\big)\psi(-\psi)]\Big|\\
\le&\Big|\Tr[\mc{R}\otimes\mc{I}(\psi)\psi\big(\mbb{I}_{q}\otimes\Tr_A(\psi)\big)]\Big|+\Big|\Tr[\big(\mbb{I}_{q}\otimes\Tr_A(\psi)\big)\psi(-\psi)]\Big|\\
\le&\norm{\mc{R}\otimes\mc{I}(\psi)}_\infty\norm{\psi}_1\norm{\Tr_A(\psi)}_\infty+\norm{\psi}_1\norm{\Tr_A(\psi)}_\infty\\
\le&2\norm{\Tr_A(\psi)}_\infty\le0.02,
\end{aligned}
\end{equation}
which implies
\begin{equation}
\begin{aligned}
\Pr[\text{get } \ket{0^{m+r+n}}]=\frac{1}{4}\Tr[\mc{R}\otimes\mc{I}(\psi)\psi\mc{R}\otimes\mc{I}(\psi)]\ge\frac{0.98}{4}=0.245.\\
\end{aligned}
\end{equation}

Therefore, when repeating the circuit shown in Fig.~\ref{fig:applications}(a) for $K$ times, if $\rho\in\mathscr{P}_0$, then $\Pr[\mathbf{0}\notin\{\mathbf{b}_1,\cdots,\mathbf{b}_K\}]=1$; if $\rho\in\mathscr{P}_0$ and $\norm{\Tr_A(\psi)}_{\infty}\le0.01$, then $\Pr[\mathbf{0}\in\{\mathbf{b}_1,\cdots,\mathbf{b}_K\}]\ge1-(1-0.245)^K=1-0.755^K$.

Combining the above derivations, the probability of success in the entanglement detection task is
\begin{equation}
\begin{aligned}
\Pr\big[\text{success}\big]=&\Pr\big[\rho\in \mathscr{P}_0,\mathbf{0}\notin\{\mathbf{b}_1,\cdots,\mathbf{b}_K\}\big]+\Pr\big[\rho\in \mathscr{P}_1,\mathbf{0}\in\{\mathbf{b}_1,\cdots,\mathbf{b}_K\}\big]\\
=&\Pr\big[\rho\in \mathscr{P}_0\big]\Pr\big[\mathbf{0}\notin\{\mathbf{b}_1,\cdots,\mathbf{b}_K\}\big|\rho\in \mathscr{P}_0\big]\\
&+\Pr\big[\rho\in \mathscr{P}_1,\mathbf{0}\in\{\mathbf{b}_1,\cdots,\mathbf{b}_K\},\norm{\Tr_A(\rho)}_{\infty}\le0.01\big]\\
&+\Pr\big[\rho\in \mathscr{P}_1,\mathbf{0}\in\{\mathbf{b}_1,\cdots,\mathbf{b}_K\},\norm{\Tr_A(\rho)}_{\infty}>0.01\big]\\
\ge&\frac{1}{2}\cdot1+\Pr\big[\rho\in \mathscr{P}_1,\mathbf{0}\in\{\mathbf{b}_1,\cdots,\mathbf{b}_K\},\norm{\Tr_A(\rho)}_{\infty}\le0.01\big]\\
=&\frac{1}{2}+\Pr\big[\rho\in \mathscr{P}_1,\norm{\Tr_A(\rho)}_{\infty}\le0.01\big]\Pr\big[\mathbf{0}\in\{\mathbf{b}_1,\cdots,\mathbf{b}_K\}\big|\rho\in \mathscr{P}_1,\norm{\Tr_A(\rho)}_{\infty}\le0.01\big]\\
\ge&\frac{1}{2}+\frac{1}{2}\cdot0.99\cdot(1-0.755^K).
\end{aligned}
\end{equation}
Therefore, for $K=2$ we have $\Pr\big[\text{success}\big]\ge2/3$.
\end{proof}

\section{Realization of $U_\mc{N}$}\label{app:sparse_BE}

Despite sample complexity, it's also important to analyze the gate complexity of BELT. The most expensive part is the realization of $U_\mc{N}$.

We say a matrix is $s$-sparse if every row and every column contains at most $s$ non-zero entries.

The following lemma from \cite{gilyn2019singular} (see also \cite{hiroka2025hardnessquantumdistributionlearning}) shows how to block encode an $s$-sparse matrix efficiently.

\begin{lemma}
Let $A\in\mbb{C}^{2^w\times 2^w}$ be an $s$-sparse matrix satisfying $\abs{A_{ij}}\le1$ for all $i,j$, and both the positions of nonzero entries and their values can be computed in polynomial time classically. Then we can implement a $(s,w+3,\epsilon)$-block encoding of $A$ using $\mc{O}\Big(\operatorname{poly}(w)+\operatorname{poly}\big(\log(s^2\epsilon^{-1})\big)\Big)$ gates and $\mc{O}\Big(\operatorname{poly}\big(\log(s^2\epsilon^{-1})\big)\Big)$ ancilla qubits.
\end{lemma}

\section{Block Encoding by Exponentiation}\label{section:BlockEncoding_by_Exponentiation}

\begin{figure*}[t]
\centering
\includegraphics[width=0.6\textwidth]{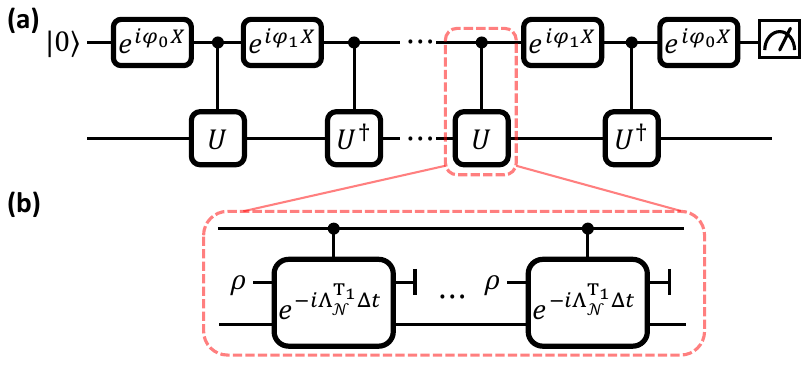}
\caption{(a) The circuit of QETU, where $U=e^{-iH}$. This circuit approximately block encodes $f(H)$. (b) For $H=\mc{N}(\rho)$, controlled-$U$ and controlled-$U^\dagger$ can be approximately realized by HME.}
\label{fig:QETU+HME}
\end{figure*}

Beyond BELT, this section presents an alternative procedure for obtaining a block encoding of $f\bigl(\mc{N}(\rho)\bigr)$ when $\mc{N}$ is a Hermitian-preserving map (rather than an arbitrary linear map, as allowed by BELT) and $f:\mathbb R\to\mathbb R$ is a real-valued function.  
The protocol combines the quantum eigenvalue transformation of unitaries with real polynomials (QETU) \cite{dong2022ground} and Hermitian-preserving map exponentiation (HME) \cite{wei2023realizing}.

\begin{lemma}[QETU \cite{dong2022ground}]\label{lemma:QETU}
Let $U=e^{-iH}$, where $H$ is an $n$-qubit Hermitian matrix.  
Consider a real function $f:\mathbb R\to\mathbb R$ and define $g(x)=2\arccos(x)$.  
If there exists a real even polynomial $F(x)$ of degree $d$ such that $\lvert F(x)\rvert\le1$ for every $x\in[-1,1]$ and 
$\sup_{x\in[\sigma_{\min},\sigma_{\max}]}\lvert(f\circ g)(x)-F(x)\rvert\le\eta$,
where $\sigma_{\min}=\cos\bigl(\lambda_{\max}(H)/2\bigr)$ and $\sigma_{\max}=\cos\bigl(\lambda_{\min}(H)/2\bigr)$,  
then a symmetric phase sequence $(\varphi_0,\varphi_1,\dots,\varphi_{d/2},\dots,\varphi_1,\varphi_0)\in\mathbb R^{d+1}$ exists such that the circuit in Fig.~\ref{fig:QETU+HME}(a) is a $(1,1,\eta)$-block encoding of $f(H)$.
\end{lemma}

\begin{lemma}[HME \cite{wei2023realizing}]\label{lemma:HME}
Let $\mc{N}$ be a Hermitian-preserving map.  
There exists a quantum algorithm that, using $\mathcal O\bigl(\epsilon^{-1}\|\Lambda_{\mc{N}}^{\mathrm{T}_1}\|_\infty^{2}t^{2}\bigr)$ copies of a quantum  state $\rho$, implements a quantum channel approximating $e^{-i\mc{N}(\rho)t}$ up to $\epsilon$ error in diamond distance.
\end{lemma}

\begin{theorem}\label{thm:BE_by_QETU_HME}
With $H=\mc{N}(\rho)$ and the other settings as in Lemma~\ref{lemma:QETU}, the protocol of Fig.~\ref{fig:QETU+HME} realizes a $(1,1,\eta)$-block encoding of $f\bigl(\mc{N}(\rho)\bigr)$ to error $\epsilon$, using
$\mc{O}\big(\epsilon^{-1}d^2\bignorm{\Lambda_\mc{N}^{\mathrm{T}_1}}_{\infty}^{2}\big)$
copies of $\rho$.
\end{theorem}
\begin{proof}
By HME, 
$\mathcal O\bigl((d^{-1}\epsilon)^{-1}\bignorm{\Lambda_\mc{N}^{\mathrm{T}_1}}_{\infty}^{2}\bigr)$ 
repetitions of the unit cell in Fig.~\ref{fig:QETU+HME}(b), and hence the same number of copies of $\rho$, suffice to implement a CPTP map that approximates each controlled-$U$ unitary in Fig.~\ref{fig:QETU+HME}(a), where $U=e^{-i\mc{N}(\rho)}$, to diamond-norm accuracy $d^{-1}\epsilon$.  
The controlled-$U^\dagger$ gates are handled in the same way.  
Overall, the protocol consumes 
$\mathcal O\bigl(\epsilon^{-1}d^{2}\bignorm{\Lambda_\mc{N}^{\mathrm{T}_1}}_{\infty}^{2}\bigr)$ 
copies of $\rho$.  
Replacing every controlled-$U$ and controlled-$U^\dagger$ with the circuits of Fig.~\ref{fig:QETU+HME}(b) adds at most $d(d^{-1}\epsilon)=\epsilon$ total approximation error \cite{watrous2018theory}.
\end{proof}

Compared to BELT, the protocol derived from Theorem~\ref{thm:BE_by_QETU_HME} offers two advantages:  
(i) it requires only copies of $\rho$, not a unitary oracle preparing a purification of $\rho$; 
(ii) only a single qubit should be measured at the end of the circuit.  
Its drawbacks are that  
(i) $\mc{N}$ must be Hermitian-preserving, and  
(ii) the HME step introduces additional approximation error, increasing the sample complexity.  
Table~\ref{tab:two_methods_comparison} presents a comparison of the two methods.

\begin{table*}[htbp]
\centering
\resizebox{0.92\textwidth}{!}{
\begin{tabular}{c|c|c|c|c}
\hline
\hline
 & Input type & Technique & Accuracy & Maps $\mc{N}$ to simulate \\
\hline
BELT
& $U^\rho$ and $(U^\rho)^\dagger$ & 
generalizing block encoding of $\rho$ \cite{Low2019hamiltonian,gilyn2019singular}
& exact & linear \\
\hline
Theorem~\ref{thm:BE_by_QETU_HME} & $\rho$ & QETU + HME & arbitrary small error $\epsilon$ & Hermitian-preserving \\
\hline
\hline
\end{tabular}
}
\caption{Comparison of two methods to block encoding $\mc{N}(\rho)$.}
\label{tab:two_methods_comparison}
\end{table*}

\end{document}